\documentclass[11pt,a4paper]{article}

\usepackage{placeins}
\usepackage{tikz}
\usepackage{tikz-3dplot}
\usepackage{tikzscale}

\usepackage{pgfplots}
\usepackage{subcaption}
\usepackage[utf8]{inputenc}
\usepackage{amsmath}
\usepackage{amssymb}
\usepackage{graphicx}
\usepackage{epstopdf}
\usepackage[amsmath,thref,thmmarks,hyperref]{ntheorem}
\usepackage{dsfont}
\usepackage{fullpage}
\usepackage{multirow}
\usepackage[affil-it]{authblk}
\usepackage{braket}

\newcommand{\ketbra}[2]{\ket{#1}\!\bra{#2}}
\usepackage{enumerate}
\usepackage[shortlabels]{enumitem}

\usepackage{mathtools}
\usepackage{nicefrac}

\usepackage[backend=biber,firstinits=true,doi=false,url=false,maxbibnames=5,style=alphabetic]{biblatex}
\renewbibmacro{in:}{}
\usepackage[toc,page]{appendix}

\usepackage[bookmarks]{hyperref}

\theoremseparator{.}
\theorempreskip{0.7\topsep}
\newtheorem{theorem}{Theorem}[section]
\newtheorem{lemma}[theorem]{Lemma}
\newtheorem{proposition}[theorem]{Proposition}
\newtheorem{corollary}[theorem]{Corollary}
\theorembodyfont{\normalfont}
\theoremstyle{plain}
\theoremsymbol{\ensuremath{\diamond}}
\newtheorem{definition}[theorem]{Definition}
\newtheorem*{remark}{Remark}
\makeatletter
\newtheoremstyle{MyNonumberplain}%
  {\item[\theorem@headerfont\hskip\labelsep ##1\theorem@separator]}%
  {\item[\theorem@headerfont\hskip\labelsep ##3\theorem@separator]}
\makeatother
\theoremstyle{MyNonumberplain}
\theoremheaderfont{\itshape}
\theorembodyfont{\upshape}
\theoremsymbol{\ensuremath{\Box}}
\newtheorem{proof}{Proof}

\usepackage{cleveref}

\usepackage{accents}
\newcommand{\ubar}[1]{\underaccent{\bar}{#1}}

\newcommand{\cA}{\mathcal{A}}
\newcommand{\cH}{\mathcal{H}}
\newcommand{\cD}{\mathcal{D}}

\newcommand{\cX}{\mathcal{X}}

\newcommand{\cF}{\mathcal{F}}
\newcommand{\cB}{\mathcal{B}}
\newcommand{\Bsa}{\mathcal{B}_\text{sa}}

\newcommand{\one}{\mathds{1}}
\newcommand{\eps}{\varepsilon}

\newcommand{\RR}{\mathbb{R}}
\newcommand{\R}{\mathbb{R}}

\newcommand{\rhomaxeps}{\rho^*_\eps}
\newcommand{\rhomineps}{\rho_{*,\eps}}

\newcommand{\pmaxeps}{p^*_\eps}
\newcommand{\pmineps}{p_{*,\eps}}

\renewcommand{\i}{\mathrm{i}}
\DeclareMathOperator{\conv}{conv}

\DeclareMathOperator{\AF}{AF}

\DeclareMathOperator{\pow}{pow}
\DeclareMathOperator{\tr}{Tr}

\DeclareMathOperator{\re}{Re}

\DeclareMathOperator{\dom}{dom}
\DeclareMathOperator{\argmin}{argmin}
\DeclareMathOperator{\argmax}{argmax}
\DeclareMathOperator{\Eig}{Eig}
\DeclareMathOperator{\diag}{diag}

\renewcommand{\epsilon}{\varepsilon}
\newcommand{\Be}{B_\varepsilon}
\newcommand{\Bep}{B_\varepsilon^+}

\renewcommand{\one}{\mathds{1}}
\renewcommand{\H}{\mathcal{H}}
\DeclareMathOperator{\spec}{spec}

\newcommand{\inv}{^{-1}}
\newcommand{\C}{\mathbb{C}}
\newcommand{\D}{\cD}
\renewcommand{\d}{\operatorname{d}\!}

\begin{document}
\title{Maximum and minimum entropy states yielding local continuity bounds}

\author{Eric P. Hanson\thanks{Email: \texttt{ephanson@damtp.cam.ac.uk}} }
\author{Nilanjana Datta\thanks{Email: \texttt{n.datta@damtp.cam.ac.uk}}}
\affil{\small Department of Applied Mathematics and Theoretical Physics, Centre for Mathematical Sciences\\ University of Cambridge, Cambridge~CB3~0WA, UK}

\maketitle
\begin{abstract}
Given an arbitrary quantum state ($\sigma$), we obtain an explicit construction of a state $\rho^*_\eps(\sigma)$ (resp.~$\rho_{*,\eps}(\sigma)$) which has the maximum (resp.~minimum) entropy among all states which lie in a specified neighbourhood ($\eps$-ball) of $\sigma$. Computing the entropy of these states leads to a local strengthening of the continuity bound of the von Neumann entropy, i.e., the Audenaert-Fannes inequality. Our bound is local in the sense that it depends on the spectrum of $\sigma$. The states $\rho^*_\eps(\sigma)$ and $\rho_{*,\eps}(\sigma)$ depend only on the geometry of the $\eps$-ball and are in fact optimizers for a larger class of entropies. These include the R\'enyi entropy and the min- and max- entropies, providing explicit formulas for certain smoothed quantities. This allows us to obtain local continuity bounds for these quantities as well. In obtaining this bound, we first derive a more general result which may be of independent interest, namely a necessary and sufficient condition under which a state maximizes a concave and G\^ateaux-differentiable function in an $\eps$-ball around a given state $\sigma$. Examples of such a function include the von Neumann entropy, and  the conditional entropy of bipartite states. Our proofs employ tools from the theory of convex optimization under non-differentiable constraints, in particular Fermat's Rule, and majorization theory.

\end{abstract}

\section{Introduction}

An important class of problems in quantum information theory concerns the determination of optimal rates of information-processing tasks such as storage and transmission of information, and entanglement manipulation. These optimal rates can be viewed as the {\em{operational quantities}} in quantum information theory. They include the data compression limit of a quantum information source, the various capacities of a quantum channel, and the entanglement cost and distillable entanglement of a bipartite state. The aim in quantum information theory is to express these operational quantities in terms of suitable {\em{entropic quantities}}. Examples of the latter include the von Neumann entropy, R\'enyi entropies, conditional entropies, coherent information and mutual information. 

Tools from the field of convex optimization theory often play a key role in the analysis of the above-mentioned quantities. This is because the relevant operational quantity of an information-processing task is often expressed as a convex optimization problem: one in which an entropic quantity is optimized over a suitable convex set. Convex optimization theory is also useful in studying various properties of the entropic quantities themselves. As we show in this paper, an important property of entropic functions which is amenable to a convex analysis is their continuity.

An explicit continuity bound for the von Neumann entropy was obtained by Fannes~\cite{Fannes1973} and later strengthened by Audenaert~\cite{Audenaert07}. It is often referred to as the Audenaert-Fannes (AF) inequality and is an upper bound on the difference of the entropies of two states in terms of their trace distance. Alicki and Fannes found an analogous continuity bound for the conditional entropy for a bipartite state \cite{AlickiFannes04}, which was later strengthened by Winter~\cite{Winter16}. Continuity bounds for R\'enyi entropies were obtained in \cite{Renyi-CMNF,Ras2011}, although neither bound is known to be sharp. These bounds are ``global'' in the sense that they only depend on the trace distance between the states (and a dimensional parameter) and not on any other property of the states in question. In \cite{Winter16}, Winter also obtained continuity bounds for the entropy and the conditional entropy
for infinite-dimensional systems under an energy constraint. His analysis was extended by Shirokov to the
quantum conditional information  for finite-dimensional systems in \cite{Shirokov15}, and for infinite-dimensional tripartite systems (under an energy constraint) in~\cite{Shirokov29}. In \cite{Shirokov15}, Shirokov also established continuity bounds for the Holevo quantity, and refined the continuity bounds on the classical and quantum capacities obtained by Leung and Smith~\cite{LS09}. Continuity bounds have also been obtained for various other entropic quantities (see e.g.~\cite{AE_relative_ent_1,KE_relative_ent_2,RW15_relative_ent_diff}
and references therein).

In this paper, we obtain a local strengthening of the continuity bounds for both the von Neumann entropy as well as the R\'enyi entropy for $\alpha \in (0,1)\cup (1,\infty)$. Given a $d$-dimensional state $\sigma$, with von Neumann entropy $S(\sigma)$, we obtain an upper bound on $|S(\rho) - S(\sigma)|$ for any state $\rho$ whose trace distance from $\sigma$ is at most $\eps$ (for a given $\eps \in (0,1)$). Our bound is {\em{local}} in the sense that it depends on the spectrum of the state $\sigma$. We prove that our bound is tighter than the Audenaert-Fannes (AF) bound and reduces to the latter only when either $\sigma$ is a pure state, $\sigma =\diag(1- \epsilon, \frac{\epsilon}{d-1},\dotsc,\frac{\epsilon}{d-1})$ when $\eps < 1 - \frac{1}{d} $, or $\sigma = \frac{1}{d}\one$ when $\epsilon\geq 1 - \frac{1}{d}$. We prove an analogous bound for the R\'enyi entropy. \Cref{fig:local-VN-vs-AF} gives a comparison of our local bound for the von Neumann entropy with the (AF)-bound, and \Cref{fig:local-Renyi-vs-global} provides a comparison of the analogous local bound for the $1/2$-R\'enyi entropy with the corresponding global bounds for R\'enyi entropies derived in~\cite{Renyi-CMNF,Ras2011}.

In order to obtain the above results, we first explicitly construct states in the $\eps-$ball, $B_\eps(\sigma)$, around $\sigma$ (i.e., the set of states which are at a trace distance of at most $\eps$ from $\sigma$) which have the maximum and minimum entropy, respectively. The construction of the maximum entropy state is described and motivated by a result from convex optimization theory. Maximizing the entropy over states in  $B_\eps(\sigma)$ can be transcribed into an optimization problem over all states involving a non-differentiable objective function. To solve the problem we employ the notion of subgradients and  {\em{Fermat's Rule}} of convex optimization, which was first described by Fermat~(see e.g. \cite{Bauschke}) in the 17th century! In fact, we first derive the following, more general, result which might be of independent interest. We consider a particular class of functions $\cF$, and derive a necessary and sufficient condition under which a state in $B_\eps(\sigma)$ maximizes any function in this class. The von Neumann entropy, the conditional entropy, and the $\alpha$-R\'enyi entropy (for $\alpha \in (0,1)$) are examples of functions in $\cF$. The precise mathematical definition of the class $\cF$ is given in \Cref{sec:classF}. The maximum entropy state is unique, and satisfies a semigroup property.

In fact, we find that the state which maximizes the entropy over $\Be(\sigma)$ can be obtained from the geometry of $\Be(\sigma)$ via majorization theory, using in particular the Schur concavity of the von Neumann entropy (more generally, this state maximizes a larger class of generalized entropies which are Schur concave). Similarly, motivated by a minimum principle for concave functions, we construct a minimum entropy state in $\Be(\sigma)$. These states provide explicit formulas for smoothed entropies\footnote{smoothed with respect to the trace distance} relevant for one shot information theory \cite{RennerThesis}.

The paper is organized as follows. We start with some mathematical preliminaries in \Cref{sec:math_prelim}. In \Cref{sec:main-results} we state our main results (see \Cref{thm:max-min-ball}, \Cref{prop:local_cont_bound} and \Cref{thm:main_convex_result}). The proof of \Cref{thm:max-min-ball} entails explicit construction of maximum and minimum entropy states in the $\eps$-ball, which are given in \Cref{sec:geometry_trace_ball}. The proof of our main mathematical result \Cref{thm:main_result} is given in \Cref{sec:prove_convex_result}. We end the paper with a Conclusion. 

\section{Mathematical tools and notation \label{sec:math_prelim}}
\subsection{Quantum states and majorization}
	We restrict our attention to finite-dimensional Hilbert spaces. For Hilbert spaces $\H_1,\H_2$, we denote the set of linear maps from $\H_1$ to $\H_2$ by $\cB(\H_1,\H_2)$, and write $\cB(\H) \equiv \cB(\H,\H)$ for the algebra of linear operators on a single Hilbert space $\H$. We denote the set of self-adjoint linear operators on $\H$ by $\Bsa(\H) \subset \H$. Upper case indices label quantum systems: for a Hilbert space $\H_A$ corresponding to a quantum system $A$, we write $\ket{\psi}
	_A\in \H_A$ and $\rho_A\in \cB(\H_A)$, and we use the notation $\H_{AB} = \H_A\otimes \H_B$. A \emph{quantum state} (or simply state) is a density matrix, i.e.~an operator $\rho_A\in \cB(\H_A)$ with $\rho_A\geq 0$ and $\tr \rho_A=1$. We denote the set of states on $\H_A$ by $\D(\H_A)$. We say a state $\rho$ is \emph{faithful} if $\rho>0$, and denote by $\D_+(\H_A)$ the set of faithful states on $\H_A$.  

	For results which only involve a single Hilbert space, $\H$, we set $\Bsa = \Bsa(\H)$, $\D = \D(\H)$, $\D_+ = \D_+(\H)$ and $d := \dim \H$. Let $\tau := \frac{1}{d}\one$ denote the completely mixed state. For any $A\in \Bsa$, let $\lambda_+ (A)$ and $\lambda_-(A)$ denote the maximum and minimum eigenvalues of $A$, respectively, and $\vec\lambda(A)$ the vector in $\R^d$ consisting of eigenvalues of $A$ in non-increasing order, counted with multiplicity. We denote the spectrum of an operator $A \in \Bsa$ by $\spec(A)$ and its kernel by $\ker A$.

	Note that $\Bsa$ is a real vector space of dimension $d^2$, which is a (real) Hilbert space when equipped with the Hilbert-Schmidt inner product $\Bsa\times \Bsa\ni (A,B) \mapsto \braket{A,B}_\text{HS} := \tr(AB)$, which induces the norm $\|A\|_2 = \sqrt{\tr(A^2)}$ for $A\in \Bsa$. We also consider the trace norm on $\Bsa$, defined by $\|A\|_1 = \tr|A|$ for $A\in \Bsa$. The fidelity between two quantum states $\rho$ and $\sigma$ is defined as 
	\[
	F(\rho,\sigma):= \|\sqrt{\rho}\sqrt{\sigma}\|_1.
	\]
	Given a quantum state $\sigma$, we define the $\eps$-ball around $\sigma$ as the closed set of states 
	\begin{equation}
	\Be(\sigma) = \{ \omega\in \D: \frac{1}{2}\| \omega - \sigma\|_1 \leq \eps \}. \label{def:Beps}
	\end{equation}
	We also define the $\eps$-ball of positive-definite states: $\Bep(\sigma) := \Be(\sigma)\cap \D_+$.
	The following inequality, which follows from  \cite[eq.~(IV.62)]{bhatia97}, will be useful.
	\begin{lemma} \label{lem:eig_upperbounded_trace-dist}
	For any $\rho,\sigma \in \cD$,
	\begin{equation} \label{eq:diag_down_in_ball_VN}
		 \|\Eig^\downarrow(\rho)-\Eig^\downarrow(\sigma)\|_1 \leq \|\rho -\sigma\|_1
		\end{equation}
		where for $A\in\Bsa$, $\Eig^\downarrow(A)$  is the diagonal matrix of eigenvalues of $A$ arranged non-increasing order. 
	\end{lemma}

	The notion of \emph{majorization} of vectors will be useful as well. Given $x\in \R^d$, write $x^\downarrow = (x^\downarrow_j)_{j=1}^d$ for the permutation of $x$ such that $x^\downarrow_1 \geq x^\downarrow_2 \geq \dotsm \geq x^\downarrow_d$. Given $x,y\in \R^d$, we say $x$ majorizes $y$, written $x \succ y$, if 
	\begin{equation} \label{def:majorize}
	 \sum_{j=1}^k x^\downarrow_j \geq \sum_{j=1}^k y^\downarrow_j \quad \forall k=1,\dotsc,d-1, \quad \text{and}\quad \sum_{j=1}^d x^\downarrow_j = \sum_{j=1}^d y^\downarrow_j.
	 \end{equation}
	 If $x \prec y \prec x$, then $x$ and $y$ are equal up to a permutation (see e.g.~\cite[p. 18]{marshall2011inequalities}).
	 A function $f: \R^d\to \R^d$ is \emph{Schur concave} if $f(x)\leq f(y)$ whenever $x\succ y$. The function $f$ is \emph{strictly Schur concave} if $f(x) < f(y)$ whenever $x\succ y$ and $x^\downarrow \neq y^\downarrow$. If $f:\R^d\to \R^d$ is symmetric and concave, then it is Schur concave; similarly if $f$ is symmetric and strictly concave, then it is strictly Schur concave \cite[p.~97]{marshall2011inequalities}.

	Given two states $\rho,\sigma\in \cD$, if $\vec\lambda(\rho) \prec \vec\lambda(\sigma)$, we write $\rho\prec\sigma$. Note if $\rho \prec \sigma \prec \rho$, then $\rho$ and $\sigma$ must have the same eigenvalues with the same multiplicities, and therefore must be unitarily equivalent. We say that $\varphi: \D \to \R$ is \emph{Schur-concave} if $\varphi(\rho)\geq \varphi(\sigma)$ for any $\rho,\sigma\in \cD$ with $\rho \prec \sigma$. If $\varphi(\rho) > \varphi(\sigma)$ for any $\rho,\sigma\in \cD$ such that  $\rho \prec \sigma$, and $\rho$ is not unitarily equivalent to $\sigma$, then $\varphi$ is \emph{strictly Schur-concave}.

	 It can be shown that $\rho \prec \sigma$ if and only if $\rho = \sum_{i} p_i U_i \sigma U_i^*$ for some $n\in \mathbb{N}$,  $p_i\geq 0$, $\sum_{i} p_i=1$, and $U_i\in \cB(\cH)$  unitary for $i=1,\dotsc,n$ \cite[Theorem 2-2]{alberti-uhlmann1982}. Hence, if $\varphi: \cD\to \R$ is unitarily invariant and concave, and $\rho\prec \sigma$ , then
	\[
	\varphi(\rho) \geq  \sum_{i\in I_L} p_i \varphi(U_i \sigma U_i^*) = \varphi(\sigma).
	\]
	Hence any unitarily-invariant concave function $\varphi$ is Schur concave. Moreover, by the same argument, if $\varphi$ is strictly concave and $\rho \prec \sigma$ is such that $\rho$ is not unitarily equivalent to $\sigma$, we have
	\begin{equation} \label{eq:strict_schur_phi}
	\varphi(\rho) > \varphi(\sigma)
	\end{equation}
	and $\varphi$ is strictly Schur concave.

\subsection{Entropies}
The von Neumann entropy of a quantum state $\rho\in \cD$ is defined as
\begin{align}\label{vN}
S(\rho) &= - \tr (\rho \log \rho),
\end{align}
which is the (classical) Shannon entropy of its eigenvalues, and a strictly concave function of $\rho$. We define all entropies with $\log$ base 2. For a bipartite state $\rho_{AB} \in {\mathcal{H}}_{AB}$, the conditional entropy $S(A|B)_\rho$ is given by
\begin{align}\label{cond-vN}
S(A|B)_\rho &= S(\rho_{AB}) - S(\rho_B),
\end{align}
where $\rho_B = \tr_A \rho_{AB}$ denotes the reduced state on the system $B$.
\smallskip

An important property of the von Neumann entropy, of particular relevance to us, is its continuity. It is described by the so-called Audenaert-Fannes (AF) bound which is stated in the following lemma \cite{Audenaert07} (see also Theorem 3.8 of \cite{PetzQITbook}). 

\begin{lemma}[Audenaert-Fannes bound] \label{lem:AF_equality}
	Let $\epsilon \in [0,1]$, and $\rho,\sigma \in \cD(\H)$ such that $\frac{1}{2}\|\rho - \sigma\|_1 \leq \epsilon$, and let $d = \dim \H$. Then
	\begin{equation}
 |S(\rho) - S(\sigma) | \leq \begin{cases}
 \epsilon \log (d-1) + h(\epsilon) & \text{if } \epsilon < 1 - \tfrac{1}{d} \\
	 \log d & \text{if } \epsilon \geq 1 - \tfrac{1}{d},
 \end{cases} \label{eq:Audenaert-Fannes_bound}
 \end{equation}
where $h(\eps) := - \eps \log \eps - (1-\eps) \log (1-\eps)$ denotes the binary entropy.

 Without loss of generality, assume $S(\rho) \geq S(\sigma)$. Then equality in \eqref{eq:Audenaert-Fannes_bound} occurs if $\sigma$ is a pure state, and either
\begin{enumerate}
	\item $\epsilon < 1 - \frac{1}{d}$ and $\vec \lambda (\rho) = (1- \epsilon, \frac{\epsilon}{d-1},\dotsc \frac{\epsilon}{d-1})$, or
	\item $\epsilon \geq 1 - \frac{1}{d}$ and $\rho= \tau:= \frac{1}{d}\one$.
\end{enumerate}

	\end{lemma}
\medskip

There are many generalizations of the von Neumann entropy. Perhaps the most general of these are the $(h,\phi)$-entropies, first studied in the quantum case by \cite{Bosyk2016}. Given $h: \R\to \R$ and $\phi: [0,1]\to \R$ with $\phi(0) = 0$ and $h(\phi(1)) = 0$, such that either $h$ is strictly increasing and $\phi$ strictly concave, or $h$ is strictly decreasing and $\phi$ strictly convex, one defines
\[
H_{(h,\phi)}(\rho) := h( \tr [\phi(\rho)])
\]
where $\phi$ is defined on $\cD$ by functional calculus, i.e., given the eigen-decomposition $\rho = \sum_i \mu_i P_i$, we have $\phi(\rho) = \sum_i \phi(\mu_i) P_i$. Particular choices of $h$ and $\phi$ yield different entropies:
\begin{itemize}
	\item For $h (x)= x$ and $\phi(x) = - x \log x$, we recover $H_{(h,\phi)} (\rho)= S(\rho)$, the von Neumann entropy.
	\item For $\alpha\in (0,1)\cup(1,\infty)$, by choosing $h(x) = \frac{1}{1-\alpha}\log(x)$ and $\phi = x^\alpha$, we find $H_{(h,\phi)}(\rho)$ reduces to the \emph{$\alpha$-R\'enyi entropy}: 
	\[
	S_\alpha(\rho) := \frac{1}{1-\alpha}\log \tr[\rho^\alpha].
	\]
	 For $\alpha\in (0,1)$, $h$ is strictly increasing, and $\phi$ is strictly concave, and for $\alpha \in (1,\infty)$, $h$ is strictly decreasing and $\phi$ strictly convex.
	\item For $\alpha \in (0,1)\cup (1,\infty)$ and $s>0$, choosing $h(x) = \frac{x^s - 1}{(1-\alpha)s}$ and $\phi(x) = x^\alpha$ yields the unified entropy $E_\alpha^{(s)}(\rho) = \frac{1}{(1-\alpha) s}[\tr(\rho^\alpha)^s-1]$. As with the R\'enyi entropies, for $\alpha\in (0,1)$, $h$ is strictly increasing, and $\phi$ is strictly concave, and for $\alpha \in (1,\infty)$ $h$ is strictly decreasing and $\phi$ is strictly concave. Note that the previous entropies are given by limits of the unified entropies \cite{KimSanders11}: $\lim_{s\to 0} E_\alpha^{(s)}(\rho) = S_\alpha(\rho)$, and $\lim_{s\to 1} E_\alpha^{(s)} = \frac{1}{1-\alpha}[\tr (\rho^\alpha) -1]$ which is called the \emph{$\alpha$-Tsallis} entropy. Lastly, $\lim_{\alpha\to 1} E_\alpha^s(\rho) = S(\rho)$ for any $s>0$.
\end{itemize}
Let us briefly summarize some of the properties of the $(h,\phi)$-entropies, as proven in \cite{Bosyk2016}:
\begin{itemize}
	\item If $\rho$ has eigenvalues $\mu_1,\dotsc,\mu_d$, counted with multiplicity, then
	\[
	H_{(h,\phi)}(\rho) = h \left( \sum_{i=1}^d \phi(\mu_i) \right)
	\]
	which is the classical $(h,\phi)$-entropy of the eigenvalues of $\rho$.
	\item Strict Schur concavity: If $\rho \prec \rho'$, then $H_{(h,\phi)} (\rho) \geq H_{(h,\phi)}(\rho')$, with equality if and only if $\rho$ is unitarily equivalent to $\rho'$. 
	\item  Bounds: $0 \leq H_{(h,\phi)}(\rho)  \leq h( d\, \phi( \tfrac{1}{d}))$. 
	\item If $h$ is concave, then $H_{(h,\phi)}$ is concave.
	\item If $\rho = \sum_{i=1}^k p_i \ketbra{\psi_i}{\psi_i}$ is an arbitrary statistical mixture of pure states, with $p_i\geq 0$ and $\sum_{i=1}^k p_i = 1$, then $H_{(h,\phi)}(\rho) \leq h \left( \sum_{i=1}^k \phi(p_i) \right)$.

\end{itemize}

We will also also consider the min- and max-entropy introduced by Renner~\cite{RennerThesis} which play an important role in one-shot information theory (see e.g.~\cite{Tom-book} and references therein):
\[
H_\text{min}(\rho) := - \log \lambda_+(\rho), \qquad H_\text{max}(\rho) := S_{1/2}(\rho) = 2 \log \tr [\sqrt{\rho}].
\]
Note $\lambda_+(\rho) = \|\rho\|_\infty$ is unitarily invariant and convex, and therefore Schur convex: if $\rho \prec \sigma$, then $\lambda_+(\rho) \leq \lambda_+(\sigma)$. Since $x\mapsto - \log (x)$ is decreasing, $H_\text{min}(\rho) \geq H_\text{min}(\sigma)$;  therefore, $H_\text{min}$ is Schur concave. As a R\'enyi entropy, $H_\text{max}$ is Schur concave as well.

The Shannon entropy of a classical random variable $X$ (taking values in a discrete alphabet ${\mathcal{X}}$) with probability mass function (p.m.f.) $(p(x))_{x \in {\mathcal{X}}}$ is given by $H(X):= - \sum_{x\in {\mathcal{X}}}p(x) \log p(x)$. The joint- and conditional entropies of two random variables $X$ and $Y$, with joint probability distribution $\left(p(x,y)\right)_{x,y}$, are respectively given by \begin{align}
H(X,Y) &:= - \sum_{x,y} p(x,y) \log p(x,y)\nonumber\\
H(X|Y) &:=H(X,Y) - H(X).
\nonumber 
\end{align}
Fano's inequality (\cite[p.~187]{Fano1961}) with equality conditions (\cite[p.~41]{han2007mathematics}), stated as \Cref{lem:Fano} below, provides an upper bound on $H(X|Y)$ and will be employed in our proofs.
\begin{lemma}[Fano's inequality] \label{lem:Fano}
For two random variables $X,Y$ which take values on $\{1,\dotsc,d\}$, we have, for $\epsilon:= \Pr(X\neq Y)$,
\[
H(X|Y)\leq \epsilon \log(d-1) + h(\epsilon),
\]
with equality if and only if 
\[
\Pr(X=i | Y = j) = \begin{cases}
1 - \epsilon & \text{if } i=j\\
\frac{\epsilon}{d-1} & \text{if } i\neq j
\end{cases}
\]
for each $j$ such that $\Pr(Y=j) \neq 0$.
\end{lemma}
\section{Main results and applications}\label{sec:main-results}
	Fix $\eps\in(0,1]$ and a state $\sigma\in \D(\H)$, with $d:=\dim \H$. Our first result is that the $\eps$-ball around $\sigma$, $\Be(\sigma)$ defined by \eqref{def:Beps}, admits a minimum and maximum in the majorization order. Note that since majorization is a partial order, a priori one does not know that there are states in $\Be(\sigma)$ comparable to every other state in $\Be(\sigma)$.
	\begin{theorem}\label{thm:max-min-ball}
	Let $\sigma\in \cD$ and $\epsilon\in (0,1]$. Then there exists two states $\rhomaxeps(\sigma)$ and $\rhomineps(\sigma)$ in $\Be(\sigma)$ (which are defined by \Cref{eq:def_rho_eps_sigma} in \Cref{sec:const_min_state_maj_order} and \Cref{eq:def_rho_eps_lowerstar} in \Cref{sec:const_max_state_maj_order} respectively) such that for any $\omega\in \Be(\sigma)$,
	\begin{equation}
	\rhomaxeps(\sigma) \prec \omega \prec \rhomineps(\sigma). \label{eq:min-max-maj-cond}
	\end{equation}
	$\rhomaxeps(\sigma)$ is the unique state in $\Be(\sigma)$ satisfying the left-hand relation of \eqref{eq:min-max-maj-cond}, and $\rhomineps(\sigma)$ is unique as an element of $\Be(\sigma)$ satisfying the right-hand relation of \eqref{eq:min-max-maj-cond} up to unitary equivalence. Furthermore, $\rhomineps(\sigma)$ lies on the boundary of $\Be(\sigma)$, and either $\rhomaxeps(\sigma) = \tau:= \frac{1}{d}\one$ or $\rhomaxeps(\sigma)$ lies on the boundary of $\Be(\sigma)$. Additionally, $\rhomaxeps(\sigma)$ satisfies the following semigroup property. If $\eps_1,\eps_2\in (0,1]$, we have
	\[
	\rho^*_{\eps_1 + \eps_2}(\sigma) = \rho^*_{\eps_1} ( \rho^*_{\eps_2}(\sigma)).
	\]
	 The state $\rhomaxeps(\sigma)$ also saturates the triangle inequality with respect to the states $\sigma$ and $\tau$, in that
	\[
	 \frac{1}{2}\|\sigma-\tau\|_1 =\frac{1}{2}\| \tau -  \rho_\eps^*(\sigma) \|_1+\frac{1}{2}\| \rho_\eps^*(\sigma) - \sigma\|_1  .
	\]
	
	\end{theorem}
\smallskip

	\paragraph{An application.} Quantum state tomography is the process of estimating a quantum state $\sigma$ by performing measurements on copies of $\sigma$. The so-called MaxEnt (or maximum-entropy) principle gives that an appropriate estimate of $\sigma$ is one which is compatible with the constraints on $\sigma$ which can be determined from the measurement results and which has maximum entropy subject to those constraints \cite{BDDAW99,Jaynes57}. 

	 Given a target pure state $\sigma_\text{target}$, one can estimate the fidelity between $\sigma$ and $\sigma_\text{target}$ efficiently, using few Pauli measurements of $\sigma$ \cite{FL11,dSLCP11}. Using the bound $\frac{1}{2}\|\sigma - \sigma_\text{target}\|_1 \leq \eps:=\sqrt{1 - F(\sigma,\sigma_\text{target})^2}$, one therefore obtains a bound on the trace distance between $\sigma$ and $\sigma_\text{target}$. \Cref{thm:max-min-ball} gives that the state with maximum entropy in $\Be(\sigma_\text{target})$ is $\rho^*_\eps(\sigma_\text{target})$. Using that $\sigma_\text{target}$ is a pure state, one can check using \Cref{lem:char_of_alpha1} that
	\begin{equation} \label{eq:def_rho_eps_sigma_target}
	\rho^*_\eps(\sigma_\text{target}) = \begin{cases}
		\diag(1-\eps, \frac{\eps}{d-1},\dotsc,\frac{\eps}{d-1})  & \text{ if } \eps \leq 1 - \frac{1}{d}\\
		\tau:=\frac{\one}{d} & \text{else}
		\end{cases}
	\end{equation}
	in the basis in which $\sigma_\text{target} = \diag(1,0,\dotsc,0)$, where $d$ is the dimension of the underlying Hilbert space. The maximum-entropy principle therefore implies that $\rho^*_\eps(\sigma_\text{target})$ defined by \eqref{eq:def_rho_eps_sigma_target} is the appropriate estimate of $\sigma$, when given only the condition that $\frac{1}{2}\|\sigma - \sigma_\text{target}\|_1 \leq \eps$.
	\begin{itemize}
		\item it may be possible to determine additional constraints on $\sigma$ by the Pauli measurements performed to estimate $F(\sigma,\sigma_\text{target})$. In that case, MaxEnt gives that the appropriate estimate of $\sigma$ is the state with maximum entropy subject to these constraints as well, and not only the relation $\frac{1}{2}\|\sigma - \sigma_\text{target}\|_1 \leq \eps$.
		\item it may be possible to devise a measurement scheme to estimate $\frac{1}{2}\|\rho-\sigma_\text{target}\|_1$ directly, which could be more efficient than first estimating the fidelity and then employing the Fuchs-van de Graaf inequality to bound the trace distance.
	\end{itemize}
	
	\medskip
	\Cref{thm:max-min-ball} immediately yields maximizers and minimizers over $\Be(\sigma)$ for any Schur concave function $\varphi$. Note that, as stated in the following corollary, the minimizer $\rhomaxeps(\sigma)$ (resp.~the maximizer $\rhomineps(\sigma)$) in the majorization order (\ref{eq:min-max-maj-cond}) is the maximizer (resp.~minimizer) of $\varphi$ over the $\eps$-ball $\Be(\sigma)$.

	\begin{corollary}\label{cor:max-min-ball-Schurconcave}
	Let $\varphi: \cD \to \R$ be Schur concave. Then
	\[
	\rhomaxeps(\sigma) \in \argmax_{\Be(\sigma)}\varphi, \qquad \text{and} \qquad \rhomineps(\sigma) \in \argmin_{\Be(\sigma)}\varphi.
	\]
	If $\varphi$ is strictly Schur concave, any other state $\rho' \in \argmax_{\Be(\sigma)} \varphi$ (resp.~ $\rho'\in \argmin_{\Be(\sigma)}\varphi$) is unitarily equivalent to $\rho^\eps_*(\sigma)$ (resp.~$\rhomineps(\sigma)$). If $\varphi$ is strictly concave, then $\argmax_{\Be(\sigma)} \varphi = \{\rhomaxeps(\sigma)\}$.
	\end{corollary}

	In particular, \Cref{cor:max-min-ball-Schurconcave} yields maximizers and minimizers of any $(h,\phi)$-entropy. This allows computation of (trace-ball) ``smoothed'' Schur-concave functions. Given $\varphi: \cD\to \R$, we define 
	\begin{align}
	\bar\varphi^{(\epsilon)}(\sigma) &:= \max_{\omega\in \Be(\sigma)} \varphi(\sigma), \qquad \ubar \varphi^{(\epsilon)}(\sigma) := \min_{\omega\in \Be(\sigma)} \varphi(\sigma).
	\label{eq-x}
\end{align}
	By \Cref{cor:max-min-ball-Schurconcave}, we therefore obtain explicit formulas: $\bar \varphi^{(\epsilon)}(\sigma) = \varphi(\rhomaxeps(\sigma))$, and $\ubar \varphi^{(\epsilon)}(\sigma) = \varphi(\rhomineps(\sigma))$. In particular, this provides an exact version of Theorem 1 of \cite{Skorski16}, which formulates approximate maximizers for the smoothed $\alpha$-R\'enyi entropy $\bar S_\alpha^{(\eps)}$.

	 Note that by setting $\varphi = H_\text{min}$ or $H_\text{max}$, the min- and max-entropies, in (\ref{eq-x}), yields explicit expressions for the min- and max- entropies smoothed over the $\eps$-ball. These choices are of particular interest, due to their relevance in one-shot information theory (see e.g.~\cite{RennerThesis, Tom-book} and references therein). In particular, let us briefly consider one-shot classical data compression. Given a source (random variable) $X$, one wishes to encode output from $X$ using codewords of a fixed length $\log m$, such that the original message may be recovered with probability of error at most $\eps$. It is known that the minimal value of $m$ at fixed $\eps$, denoted $m_*(\eps)$, satisfies
	 \begin{equation}
	  \ubar H_\text{max}^{(\eps)}(X) \leq \log m_*(\eps) \leq \inf_{\delta\in(0,\eps)} [ \ubar H_\text{max}^{(\eps)}(X) + \log \frac{1}{\delta}]
	 \end{equation}
	 as shown by \cite{Tomamichel2016,RR12}. Equation~\eqref{eq-x} provides the means to explicitly evaluate the quantity $ \ubar H_\text{max}^{(\eps)}(X)$ in this bound.

\begin{remark}
Interestingly, the states $\rho^*_\eps(\sigma)$ and  $\rho_{*,\eps}(\sigma)$ were derived independently by  Horodecki and Oppenheim~\cite{MHJO}, and in the context of thermal majorization by van der Meer and Wehner~\cite{Remco}. They referred to these as the flattest and steepest states. Horodecki and Oppenheim also used these states to obtain expressions for smoothed Schur concave functions.
\end{remark}

	By applying \Cref{cor:max-min-ball-Schurconcave} to the von Neumann entropy, we obtain a local continuity bound given by inequality (\ref{eq:local_cont_bound}) in the following proposition. We also compare it to the Audenaert-Fannes bound (stated in \Cref{lem:AF_equality}), which we include as inequality (\ref{eq:AF_from_local}) below. In addition, we establish that the sufficient condition for equality in the (AF)-bound (see \Cref{lem:AF_equality}) is also a necessary one. 
	\begin{proposition}[Local continuity bound] \label{prop:local_cont_bound}
	Let $\sigma\in \D$ and $\eps\in (0,1]$. Then for any state $\omega \in \Be(\sigma)$,
	\begin{align}	
	|S(\omega) - S(\sigma)| &\leq \max \{  S(\rho_\eps^*(\sigma)) - S(\sigma), S(\sigma) - S(\rho_{*,\eps}(\sigma)) \}\label{eq:local_cont_bound}\\
	&\leq \begin{cases}
 \epsilon \log (d-1) + h(\epsilon) & \text{if } \epsilon < 1 - \tfrac{1}{d} \\
	 \log d & \text{if } \epsilon \geq 1 - \tfrac{1}{d}.
 \end{cases} \label{eq:AF_from_local}
	\end{align}
	for $h(\eps) = - \eps\log \eps - (1 -\eps)\log (1- \eps)$ the binary entropy. Moreover, equality holds in \eqref{eq:AF_from_local} if and only if one of the following holds:
	\begin{enumerate}
		\item $\sigma$ is a pure state; in this case $\rho_\eps^*(\sigma) = \begin{cases}
		\diag(1- \epsilon, \frac{\epsilon}{d-1},\dotsc \frac{\epsilon}{d-1}) & \text{if } \eps < 1 - \frac{1}{d},\\
		\one/d & \text{if } \eps \geq 1 - \frac{1}{d}.
		\end{cases}$
		\item $\epsilon < 1 - \frac{1}{d}$, and $\sigma = \diag(1- \epsilon, \frac{\epsilon}{d-1},\dotsc \frac{\epsilon}{d-1})$; in this case $\rho_{\eps, *}(\sigma)$ is a pure state,
		\item $\epsilon \geq 1 - \frac{1}{d}$ and $\sigma= \tau:= \one/d$; in this case $\rho_{\eps, *}(\sigma)$ is a pure state.
	\end{enumerate}
	\end{proposition}
\begin{remark}
See \Cref{fig:local-VN-vs-AF} for a comparison between the right-hand side of our bound \eqref{eq:local_cont_bound} and the right-hand side of the (AF)-bound \eqref{eq:AF_from_local} for 500 random choices of $\sigma \in \cD$ and $\epsilon\in (0,1]$. The figure shows the merit of the local continuity bound.
\end{remark}
\begin{figure}[ht]
	\begin{center}
	\centering
	\includegraphics{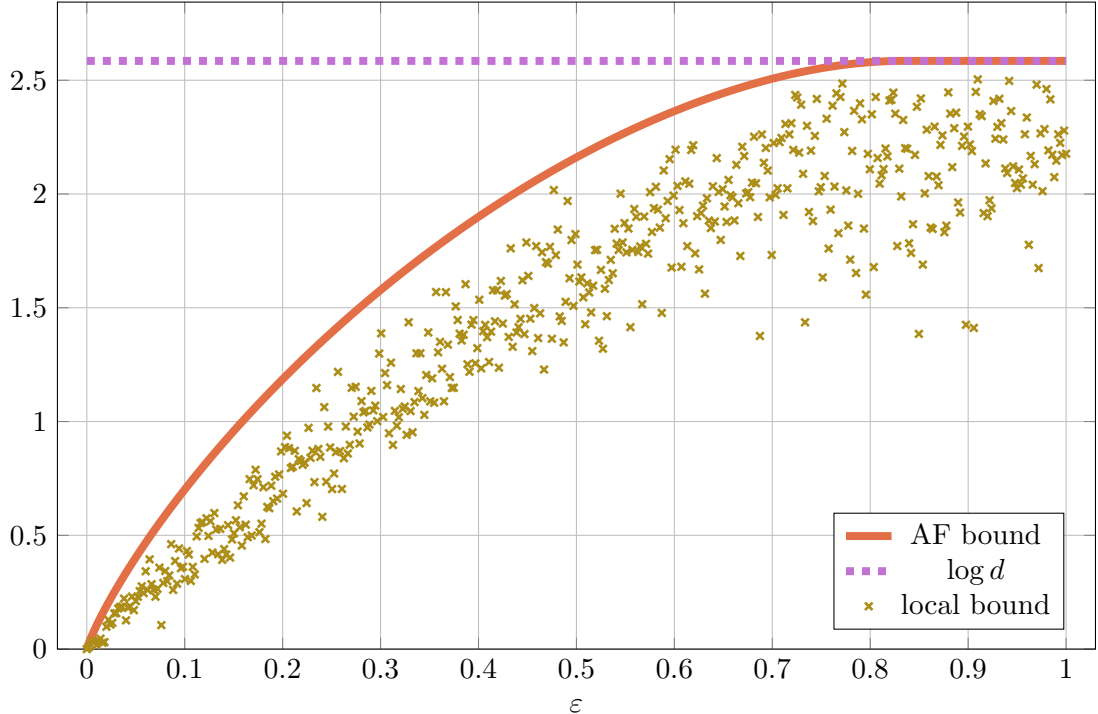}
	\end{center}
	\caption{In dimension $d=6$, the eigenvalues of 500 quantum states $\sigma$ and $\eps\in (0,1]$ were chosen uniformly randomly. For each pair $(\sigma, \eps)$, the local bound given in the right-hand side of \eqref{eq:local_cont_bound} is plotted as a cross. The Audenaert-Fannes bound given by the right-hand side of \eqref{eq:AF_from_local} is plotted as a function of $\eps$.  \label{fig:local-VN-vs-AF}} 
	\end{figure}

We can similarly find a local continuity bound for the R\'enyi entropies. This is given by the following proposition (whose proof is analogous to that of \Cref{prop:local_cont_bound}). We compare our bound \eqref{eq:local_cont_bound_Renyi} to the global bound obtained by Rastegin~\cite{Ras2011} (see also~\cite{Renyi-CMNF}), which is given by
the inequality (\ref{eq:Ras_Renyi_bound}) below.
\begin{proposition}[Local continuity bound for R\'enyi entropies] \label{prop:local_cont_bound_Renyi}
	Let $\alpha\in(0,1)\cup(1,\infty)$. Let $\sigma\in \D$ and $\eps\in [0,1]$. Then for any state $\omega \in \Be(\sigma)$,
	\begin{align}	
	|S_\alpha(\omega) - S_\alpha(\sigma)| &\leq \max \{  S_\alpha(\rho_\eps^*(\sigma)) - S_\alpha(\sigma), S_\alpha(\sigma) - S_\alpha(\rho_{*,\eps}(\sigma)) \}\label{eq:local_cont_bound_Renyi}\\
	&\leq \begin{cases}
	(2 \epsilon)^\alpha g_\alpha(d) + \eta_\alpha(2 \epsilon) &  \alpha < 1 \text{ and }  \epsilon < \alpha^{1/(1-\alpha)} \\
	d^{2(\alpha-1)} [g_\alpha(d-1) + r_\alpha(\epsilon)] & \alpha > 1.
	\end{cases} \label{eq:Ras_Renyi_bound}
	\end{align}
	where $r_\alpha( \epsilon) := c_\alpha [\epsilon^\alpha + (1 - \epsilon)^\alpha -1]$, $g_\alpha(x) := c_\alpha(x^{1-\alpha} - 1) $, and $\eta_\alpha(x) :=c_\alpha (x^\alpha - x)$, where $c_\alpha := [\ln(2)(1-\alpha)]\inv$.
	\end{proposition}

	Our characterization of maximum entropy states originates from the following theorem, which is a condensed form of \Cref{thm:main_result}, our main mathematical result.
	Given a suitable function $\varphi:\cD(\cH) \to \R$,
	and any state $\sigma \in \cD(\cH)$, the theorem provides a necessary and sufficient condition under which a state maximizes $\varphi$ in the $\eps$-ball (of positive definite states), $\Bep(\sigma)$, of the state $\sigma$. 

	\begin{theorem} \label{thm:main_convex_result}
	Let $\sigma\in \cD(\cH)$, $\epsilon \in (0,1]$, and $\varphi:\cD(\cH) \to \RR$ be a concave, continuous function which is G\^ateaux-differentiable\footnote{For the definition of Gateaux-differentiability and Gateaux gradient see \Cref{tools-convex}.} on $\cD_+(\cH)$.  A state $\rho\in \Bep(\sigma)$, satisfies $\rho\in \argmax_{\Be(\sigma)} \varphi$ if and only if \emph{both} of the following conditions are satisfied. Here $L_\rho := \nabla \varphi(\rho)$ denotes the G\^ateaux-gradient of $\varphi$ at $\rho$.
		\begin{enumerate}
		\item Either $\frac{1}{2}\|\rho-\sigma\|_1= \eps$ or $L_\rho = \lambda \one$ for some $\lambda \in \R$, and
		\item  we have 
		\begin{align}\label{eq-eval}
		\pi_\pm L_\rho \pi_\pm &= \lambda_\pm(L_\rho) \pi_\pm,
		\end{align}
where $\pi_\pm$ is the projection onto the support of $\Delta_\pm$, and where $\Delta= \Delta_+ - \Delta_-$ is the Jordan decomposition of $\Delta:= \rho-\sigma$.
	\end{enumerate}
	\end{theorem}

\smallskip

	\begin{figure}[ht]
	\centering
	\includegraphics{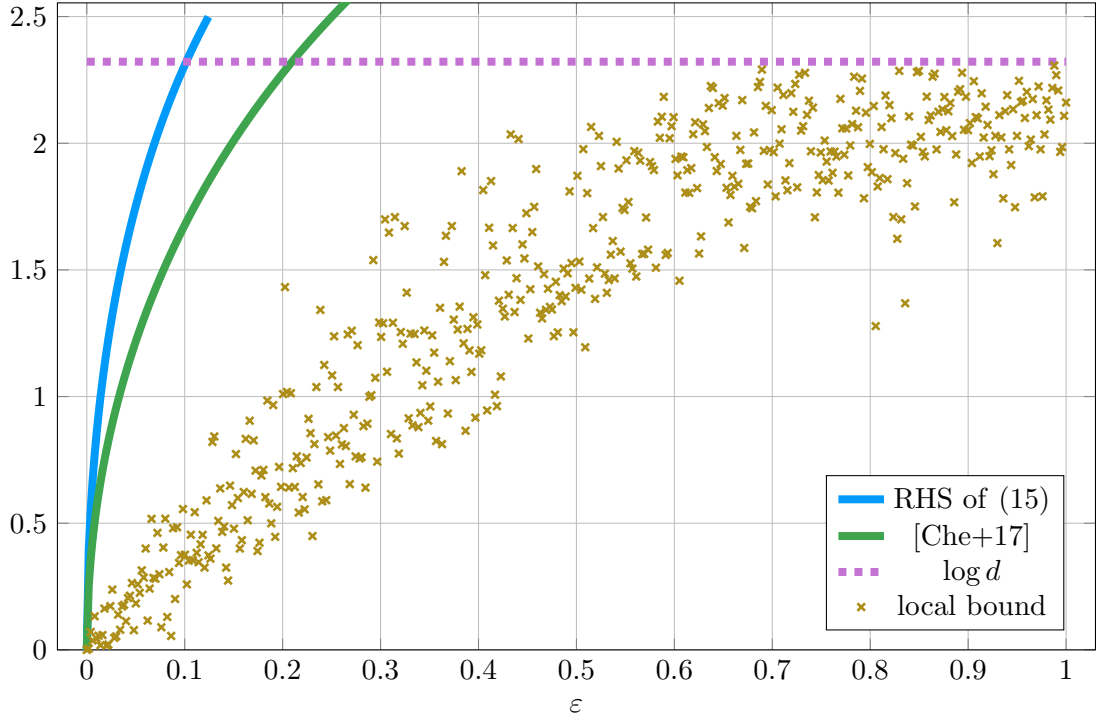}
	\caption{
	Local bounds versus the global bound for $H_\text{max}$. In dimension $d=6$, the eigenvalues of 500 quantum states $\sigma$ and $\eps\in (0,1]$ were chosen uniformly randomly. For each pair $(\sigma, \eps)$, the local bound given in the right-hand side of \eqref{eq:local_cont_bound_Renyi} is plotted with a cross. The global bound given by the right-hand side of \eqref{eq:Ras_Renyi_bound} is plotted in blue as a function of $\eps \in (0,\frac{1}{8})$, and the global bound of \cite{Renyi-CMNF} in green for $\eps\in(0,1)$.  The trivial bound $\log d$ is plotted as a dashed line. \label{fig:local-Renyi-vs-global}} 
	\end{figure}

A corollary of this theorem concerns the conditional entropy $S(A|B)_\rho$ of a bipartite state $\rho_{AB}$. It corresponds to the choice $\varphi(\rho_{AB}) = S(A|B)_\rho$ whose G\^ateaux derivative $L_\rho := \nabla \varphi(\rho_{AB})$ is given by $L_{\rho_{AB}} = - (\log \rho_{AB} -\one_A \otimes \log \rho_B) $ (see \Cref{cor:Lrho_CE}).

	\begin{corollary} \label{cor:main_thm_for_CE}
	Given a state $\sigma_{AB}\in \cD_+(\cH_A\otimes\cH_B)$ and $\eps\in (0,1]$, a state $\rho_{AB}\in \Bep(\sigma_{AB})$ has maximum conditional entropy if and only if
	\[
	S(A|B)_\rho - S(A|B)_\sigma + D(\sigma_{AB}\|\rho_{AB}) - D(\sigma_B\|\rho_B) = \eps (\lambda_+(L_\rho) - \lambda_-(L_\rho))
	\]
	where $L_\rho = \one_A \otimes \log \rho_B-\log \rho_{AB}  $.

	\end{corollary}

	\FloatBarrier

\section{Geometry of the $\eps$-ball} \label{sec:geometry_trace_ball}
In this section, we consider the $\eps$-ball $\Be(\sigma)$ around a state $\sigma \in \D(\H)$, and motivate the construction of maximal and minimal states in the majorization order given in (\ref{eq:min-max-maj-cond}) of \Cref{sec:const_min_state_maj_order}. We prove \Cref{thm:max-min-ball} by reducing it to the classical case of discrete probability distributions on $d$ symbols, and then constructing explicit states $\rhomaxeps(\sigma)$ (in \Cref{sec:const_min_state_maj_order}), and $\rhomineps(\sigma)$ (in \Cref{sec:const_max_state_maj_order}), whose eigenvalues are respectively given by the probability distributions which are minimal and maximal in majorization order.

The reduction to the classical case is immediate: the state majorization $\rho\prec \sigma$ by definition means that the eigenvalue majorization $\vec \lambda(\rho)\prec \vec \lambda(\sigma)$ holds. Thus, instead of the set of density matrices $\cD$, we consider the simplex of probability vectors
 \[
 \Delta := \{ r = (r_i)_{i=1}^d \in \R^d: r_j\geq 0 \text{ for each } j=1,\dotsc,d, \text{ and } \sum_{i=1}^d r_i = 1\}.
 \]
 Note that $\Delta$ is the polytope (i.e.~the convex hull of finitely many points) generated by $(1,0,\dotsc,0)$ and its permutations. Instead of the ball $\Be(\sigma)$, we consider the $1$-norm ball around $q = (q_i)_{i=1}^d := \vec \lambda (\sigma)$,
 \[
 \Be(q) := \{p = (p_i)_{i=1}^d\in \Delta:  \frac{1}{2}\|p-q\|_1 := \frac{1}{2}\sum_{i=1}^d |p_i - q_i| \leq \eps\}.
 \] 
The set $\{ x \in \R^d: \|x\|_1 \leq 1\}$ can be written
\[
\{ x \in \R^d: \|x\|_1 \leq 1\} = \conv\{e_1,-e_1,\dotsc,e_d,-e_d\},
\]
where  $\conv(A)$ denotes the convex hull of a set $A$, and $e_1,\dotsc,e_d$ are the vectors of the standard basis (e.g. $e_j = (0,\dotsc,0,1,0,\dotsc,0)$ with $1$ in the $j$th entry), and is therefore a polytope, called the \emph{$d$-dimensional cross-polytope} (see e.g.~\cite[p.~82]{matousek2002lectures}). As a translation and scaling of the $d$-dimensional cross-polytope, the set $\{p \in \R^d: \frac{1}{2}\|p -q\|_1 \leq \eps\}$ is a polytope as well. As $B_\epsilon(q)$ is the intersection of this set and $\Delta$, it too is a polytope (see \Cref{fig:Mwa} for an illustration of $\Be(q)$ in a particular example).

The existence of $\rhomaxeps(\sigma)$ and $\rhomineps(\sigma)$  in $\Be(\sigma)$ satisfying \eqref{eq:min-max-maj-cond} is equivalent to $\pmaxeps(q)$ and $\pmineps(q)$ in $\Be(q)$ satisfying
\begin{equation}\label{class-maj}
\pmaxeps(q) \prec w \prec \pmineps(q) 
\end{equation}
for any $w\in \Be(q)$. Using Birkhoff's Theorem (e.g.~\cite[Theorem 12.12]{NC}), the set of vectors majorized by $w \in \Delta$ can be shown to be given by 
	\begin{equation}
	M_{w}:=\{p\in \Delta: w \succ p\} = \conv \{ \pi (w): \pi \in S_d\}, \label{eq:def_Mw_char_conv_hull_of_perm}
	  \end{equation} where $S_d$ is the symmetric group on $d$ letters (see \cite[p.~34]{marshall2011inequalities}). Let us illustrate this with an example in $d=3$. Let us choose $q = \left(0.21, 0.24, 0.55\right)$ and $\epsilon = 0.1$. The simplex $\Delta$ and ball $\Be(q)$ are depicted in \Cref{fig:Mwa}. A point $w = \left(0.14, 0.28, 0.58\right) \in \Be(q)$ is shown in \Cref{fig:Mwb}, and the set $M_w$ in \Cref{fig:Mwc}. 

\begin{figure}[ht]
\centering
\begin{minipage}[b]{.3\linewidth}
\centering \includegraphics{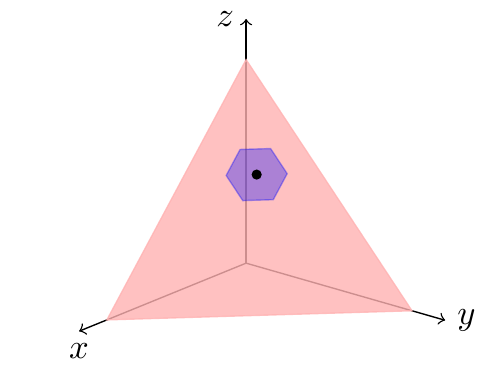}
\subcaption{$B_\epsilon(q)$, with $q$ in black}\label{fig:Mwa}
\end{minipage}\qquad
\begin{minipage}[b]{.3\linewidth}
\centering \includegraphics{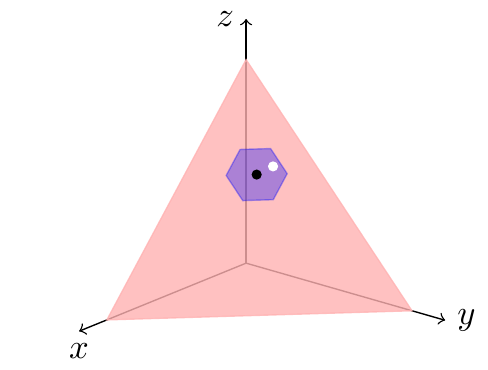}
\subcaption{A point $w\in \Be(q)$ in white.}\label{fig:Mwb}
\end{minipage}\qquad
\begin{minipage}[b]{.3\linewidth}
\centering \includegraphics{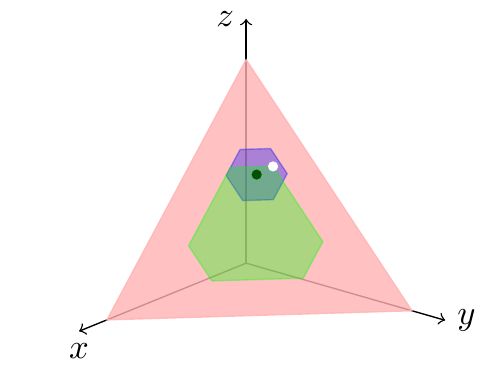}
\subcaption{The set $M_w$ in green.}\label{fig:Mwc}
\end{minipage}%
\caption{In dimension $d=3$, the simplex, $\Delta$, of probability vectors is the shaded triangle shown in (a), along with the ball $B_\epsilon(q)$ which is the hexagon shown in blue, centered at $q = \left(0.21, 0.24, 0.55\right)$ (depicted by a black dot) with  $\epsilon = 0.1$. In (b), a point $w = \left(0.14, 0.28, 0.58\right) $ is depicted in white, and in (c), the set $M_w$ is shown in green. \label{fig:Mw}}
\end{figure}
The geometric characterization \eqref{eq:def_Mw_char_conv_hull_of_perm}, depicted in \Cref{fig:Mw}, requires that $\Be(q) \subseteq M_{\pmineps(q)}$, and conversely, $\pmaxeps(q) \in M_p$ for each $p\in \Be(q)$.  \Cref{fig:Mwc} shows that for the point $w$, $\Be(q) \not \subseteq M_w$, implying that $w\neq \pmineps(q)$. Moreover, one can check that that e.g. $w \not \in M_q$, and hence $w\neq \pmaxeps(q)$.

Next we consider Schur concave functions on $\Delta$, in order to gain insight into the probability distributions $\pmineps(q)$ and $\pmaxeps(q)$ which arise
in the majorization order (\ref{class-maj}). In particular, let us consider Shannon entropy $H(p):= -\sum_{i=1}^d p_i \log p_i$ of a probability distribution $p = (p_i)_{i=1}^d$. It is known to be strictly Schur concave. Hence, if $\pmineps(q) \in \Be(q)$ satisfying \eqref{class-maj} exists, it must satisfy
\[
H(\pmineps(q)) \leq H(w)
\]
for any $w\in \Be(q)$. Thus, $\pmineps(q)$ must be a minimizer of $H$, which is a concave function, over $\Be(q)$, a convex set. Similarly, $\pmaxeps(q)$ must be a maximizer of $H$ over $\Be(q)$. Properties of maximizers of  concave functions over a convex sets are  well-understood; in particular, any local maximizer is a global maximizer.

The task of minimizing a concave function over a convex set is a priori more difficult; in particular, local minima need not be global minima. There is, however, a {\em{minimum principle}} which asserts that the minimum occurs on the boundary of the set; this is formulated more precisely in e.g. \cite[Chapter 32]{Rockafellar_book}.  Since $\Be(q)$ is a polytope,  $H$ is minimized on one of the finitely many verticies of $\Be(q)$. This fact yields a simple solution to the problem of minimizing $H$ over $\Be(q)$, as described below by example, and in generality in \Cref{sec:const_max_state_maj_order}.

Let us return to the example of \Cref{fig:Mw}. We see $\Be(q)$ has six vertices; these are $\{ q + \pi((\epsilon,- \epsilon,0)): \pi \in S_d\}$. 
\begin{figure}[ht]
\centering
\begin{minipage}[b]{.46\linewidth}
\centering \includegraphics{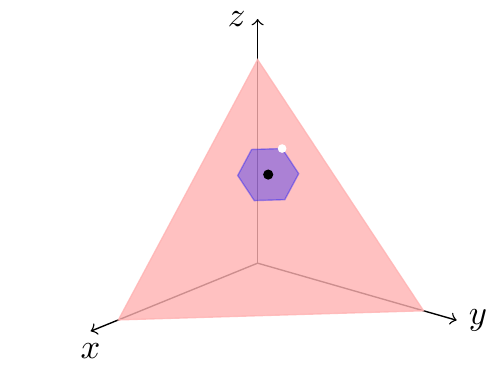}
\subcaption{A minimum $v$ of $H$ over $\Be(q)$, in white.}\label{fig:min-max-exampleb}
\end{minipage}\qquad
\begin{minipage}[b]{.46\linewidth}
\centering \includegraphics{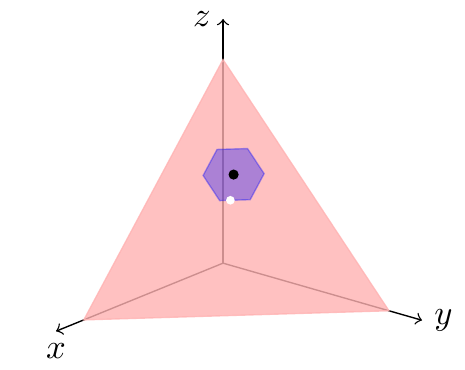}
\subcaption{The maximum $m$ of $H$ over $\Be(q)$, in white.}\label{fig:min-max-exampleb}
\end{minipage}%
 \caption{For the example of \Cref{fig:Mw}, the (unique) maximum and minimum of the Shannon entropy $H$ over $\Be(q)$ are shown. Both $v$ and $m$ occur on the boundary of $\Be(q)$.
 }\label{fig:min-max-example}
\end{figure}
The vertex which minimizes $H$ is $v:=\left(0.21- \epsilon, 0.24, 0.55+ \epsilon\right)$, where the smallest entry is decreased and the largest entry is increased, as shown in \Cref{fig:min-max-exampleb}. Moreover, one can check that $w \prec v$ for any $w\in \Be(q)$.  This leads us to the conjecture that the vertex corresponding to decreasing the smallest entry and increasing the largest entry will yield $\pmineps(q)$ satisfying \eqref{class-maj} in general. We see in \Cref{sec:const_max_state_maj_order} that this is indeed true, although in some cases more than one entry needs to be decreased.

On the other hand, finding the probability distribution $\pmaxeps(q)$ in $\Be(q)$ which satisfies \eqref{class-maj} is more than a matter of checking the verticies of $\Be(q)$, as shown by \Cref{fig:min-max-exampleb}: in this example, $\pmaxeps(q)$ is not a vertex of $\Be(q)$. Interestingly, useful insight into this probability distribution can be obtained by using results from convex optimization theory. This is discussed in the following section.

\subsection{Constructing the minimal state in the majorization order \eqref{eq:min-max-maj-cond}\label{sec:const_min_state_maj_order}}

Let us assume \Cref{thm:main_convex_result} holds, and that the von Neumann entropy $S$ satisfies the requirements of the function $\varphi$ of the theorem, with $L_\rho = - \log \rho - \tfrac{1}{\log_{\mathrm{e}}(2)}\one$. The proof of these facts are given in \Cref{sec:convexity}. Using \Cref{thm:main_convex_result}, we deduce properties of and the form of a maximizer of $S$ in the $\eps$-ball. 

Since $S$ is continuous and $\Be(\sigma)$ is compact, $S$ achieves a maximum over $\Be(\sigma)$. Moreover, since $S$ is strictly concave, the maximum is unique; otherwise, if $\rho_1,\rho_2\in \Be(\sigma)$ were maximizers, $\rho = \frac{1}{2}\rho_1 + \frac{1}{2}\rho_2\in \Be(\sigma)$ would have strictly higher entropy. Let $\rho$ be the maximizer. Condition 1 of \Cref{thm:main_convex_result} yields that either $\log \rho \propto \one$, so $\rho = \tau = \frac{\one}{d}$, or else $\frac{1}{2}\|\rho - \sigma\|_1 = \epsilon$. Since $\tau$ is the global maximizer of $S$ over $\cD$, we have $\rho = \tau$ whenever $\tau\in \Be(\sigma)$. If $\tau\not\in \Be(\sigma)$, then this condition yields the first piece of information about the maximizer: it is on the boundary of $\Be(\sigma)$, in that $\frac{1}{2}\|\rho - \sigma\|_1 = \epsilon$, just as shown in \Cref{fig:min-max-exampleb} in the classical setup.

By \Cref{lem:eig_upperbounded_trace-dist}, working in the basis in which $\sigma = \Eig^\downarrow(\sigma)$, we have 
\[
\frac{1}{2}\|\Eig^\downarrow(\rho) - \sigma\|_1 \leq \frac{1}{2}\|\rho-\sigma\|_1 \leq \epsilon
\]
and therefore $\Eig^\downarrow(\rho) \in \Be(\sigma)$. Since $S$ is unitarily invariant, $S(\Eig^\downarrow(\rho)) = S(\rho)$, and by uniqueness of the maximizer, we have $\rho = \Eig^\downarrow(\rho)$. Hence, the maximizer $\rho$ commutes with $\sigma$, and hence with $\Delta$, and the sums of its eigenprojections $\pi_\pm$. Since $[L_\rho,\rho]=0$ as well, \Cref{thm:main_convex_result} yields
\begin{equation}
L_\rho \pi_\pm = \lambda_\pm(L_\rho) \pi_\pm \label{eq:VN_L_rho_motivation}
\end{equation}
For any $\psi \in \pi_\pm \cH$, we have $L_\rho \psi = \lambda_\pm(L_\rho) \psi$, so \eqref{eq:VN_L_rho_motivation} is an eigenvalue equation for $L_\rho$. Since $\rho = \exp\left(-(L_\rho +\one)\right)$ is a function of $L_\rho$, it shares the same eigenprojections. In particular, 
\[
 \rho \pi_\pm = \exp\left(- (\lambda_{\pm}(L_\rho) +1) \right)\pi_\pm
 \] serves as an eigenvalue equation for $\rho$. Since $[\rho,\sigma]=0$, we can discuss how each acts on each (shared) eigenspace. By definition, $\rho$ and $\sigma$ act the same on $\ker \Delta = \ker(\rho-\sigma)$. On the other hand, on the subspaces where the eigenvalues of $\rho$ are greater than those of $\sigma$, i.e. on $\pi_+\cH$, we see that $\rho$ has the constant eigenvalue $\alpha_1:= \exp\left(- (\lambda_{+}(L_\rho) +1)\right)$, and on $\pi_-\cH$, $\alpha_2:= \exp\left(- (\lambda_{-}(L_\rho) +1)\right)$. Note that as $x\mapsto - \log x - 1$ is monotone decreasing on $\R$,  the subspace $\pi_+\cH$ where $L_\rho$ has its largest eigenvalue, $\rho$ has its smallest eigenvalue, and vice-versa. That is, $\lambda_+(\rho) =\alpha_2$ and occurs on the subspace $\pi_-\cH$, and $\lambda_-(\rho) = \alpha_1$, and occurs on $\pi_+\cH$.

Let us summarize the above observations. In $\ker \Delta$, the maximizer $\rho$ has the same eigenvalues as $\sigma$. On the subspace $\pi_-\cH$, $\rho$ has the constant eigenvalue $\alpha_2$, which is the largest eigenvalue of $\rho$, and $\rho \pi_- = \alpha_2 \pi_- \leq \sigma \pi_-$. In the subspace  $\pi_+\cH$, $\rho$ has the constant eigenvalue $\alpha_1$, which is its smallest eigenvalue, and $\rho \pi_+ = \alpha_1 \pi_+ \geq \sigma \pi_+$. It remains to choose subspaces corresponding to $\pi_\pm$, and the associated eigenvalues $\alpha_1$ and $\alpha_2$. Recall that as $\frac{1}{2}\|\rho-\sigma\|_1 = \epsilon$, we have $\tr[(\rho-\sigma)_+]=\tr[(\rho-\sigma)_-] = \epsilon$.

As the entropy is minimized on pure states and maximized on the completely mixed state $\tau := \frac{1}{d}\one$, one can guess that to increase the entropy, one should raise the small eigenvalues of $\sigma$, and lower the large eigenvalues of $\sigma$. That is, $\pi_+$ should correspond to the eigenspaces of the $n$ smallest eigenvalues of $\sigma$, and $\pi_-$ should correspond to the eigenspaces of the $m$ largest eigenvalues of $\sigma$, for some $n,m\in \{1,\dotsc,d-1\}$. Moreover, as $\alpha_2$ is the largest eigenvalue of $\rho$, and $\alpha_1$ is the smallest one, we must have $\alpha_1 \leq \mu \leq \alpha_2$ for any eigenvalue $\mu$ of $\sigma$ with corresponding eigenspace which is a subspace of $\ker \Delta$. \Cref{fig:rholevels} illustrates these ideas in an example. In \Cref{lem:char_of_alpha1}, we prove there exists unique $\alpha_1,\alpha_2,n$ and $m$ which respect these considerations. 

\begin{figure}[ht]
\centering
\includegraphics{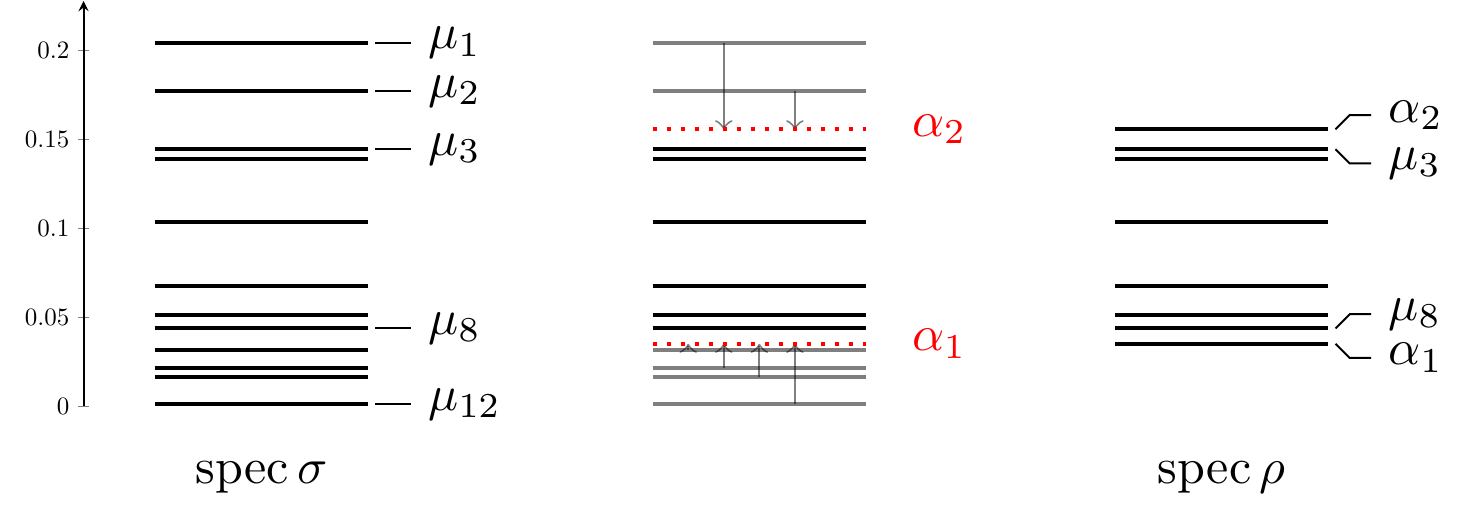}
\caption{We choose $d=12$, a state $\sigma \in \cD$, and $\eps=0.07$. Left: the eigenvalues $\mu_1 \geq \mu_2\geq \dotsc \geq \mu_d$ of $\sigma$ were plotted. Center: the smallest four eigenvalues of $\sigma$ are increased to a value $\alpha_1$, and the largest two eigenvalues of $\sigma$ decreased to $\alpha_2$, such that $\sum_{i=1}^2 (\mu_i - \alpha_2)= \sum_{j=d-3}^d (\mu_j - \alpha)= \epsilon$. Right: the eigenvalues of $\rho$ are $\alpha_2$ with multiplicity two, $\mu_3,\mu_4,\dotsc,\mu_{d-4}$, and $\alpha_1$ with multiplicity four. \label{fig:rholevels}}
\end{figure}
 
\FloatBarrier

Let us recall the task is to construct a state $\rho^*_\eps(\sigma)$ satisfying \eqref{eq:min-max-maj-cond}. Instead, we have constructed a state in order to maximize the entropy $S$. However, since any state $\rho^*_\eps(\sigma)$ satisfying \eqref{eq:min-max-maj-cond} must maximize $S$ over $\Be(\sigma)$, and the maximizer is unique, we instead check that the state resulting from this construction satisfies \eqref{eq:min-max-maj-cond} in \Cref{sec:max_of_pmaxeps}.

	Fix $\eps\in (0,1]$. Let the eigen-decomposition of a state $\sigma \in {\mathcal{D}}(\mathcal{H})$, with $\rm{dim {\mathcal{H}}} =d$, be
	$\sigma = \sum_{i=1}^d \mu_i \ketbra{i}{i}$, with 
	\[
	 \mu_1\geq \mu_2\geq \dotsm \geq \mu_d.
	 \]
	We give an explicit construction of the maximizer $\rho^* \in \argmax_{\rho\in\Be(\sigma)}S(\rho)$ as follows. We delay proof of the following lemma until \Cref{sec:proofs_VN} for readability.

	\begin{lemma} \label{lem:char_of_alpha1}\label{lem:char_of_alpha2}
	Assume $\frac{1}{2}\|\sigma-\tau\|_1 > \eps$, where $\tau = \one/d$, i.e.~$\tau \not\in B_\eps(\sigma)$. There is a unique pair $(\alpha_1,n) \in [0,1]\times \{1,\dotsc,d-1\}$ such that
	\begin{equation}
	\sum_{i=d-n+1}^{d} | \alpha_1 - \mu_i| = \eps \quad\text{and}\quad \mu_{d-n+1} < \alpha_1 \leq \mu_{d-n}, \label{eq:def_alpha1_n1}
	\end{equation}
	and similarly,  a unique pair $(\alpha_2,m) \in [0,1]\times \{1,\dotsc,d-1\}$ with
	\begin{gather}
	\sum_{j = 1}^m |\alpha_2 - \mu_i| = \eps \quad\text{and}\quad  \mu_{m+1} \leq \alpha_2  < \mu_{m}. \label{eq:def_alpha2_m}
	\end{gather}
	We define the index sets
	\begin{equation} \label{eq:def_IH_IM_IL}
	\begin{aligned}	
	I_H &= \{1,\dotsc,m \},  & I_M&=\{m+1,\dotsc,d-n \}, & I_L &= \{d-n+1,\dotsc,d\}
	\end{aligned}
	\end{equation}
	corresponding to the ``highest'' $m$ eigenvalues, ``middle'' $d-n-m$ eigenvalues, and ``lowest'' $n$ eigenvalues of $\sigma$, respectively.

	We have the following properties. The numbers $\alpha_1$ and $\alpha_2$ are such that $\alpha_1 < \frac{1}{d}< \alpha_2$, and we have the following characterizations of the pairs $(\alpha_1,n)$ and $(\alpha_2,m)$: For any $n \in {\mathbb{N}}$,defining
	\begin{align}\label{eq-alpha}
	\alpha_1(n) :=\frac{1}{n}\left(  \sum_{j\in I_L}\mu_j + \eps \right),& \qquad \alpha_2(n):= \frac{1}{n}\left( \sum_{j\in I_H}\mu_j - \eps \right),
	\end{align}
	we have $\alpha_1 = \alpha_1(n)$ and $\alpha_2 = \alpha_2(m)$, where $n$ and $m$ are, respectively, the unique solutions of the following:
	\begin{gather}
	 \mu_{d-n'+1} < \alpha_1 \leq \mu_{d-n'}\quad:\quad n' \in \{1,\dotsc,d-1\}, \label{eq:n1_as_alpha1(n1)_inequality} \\ 
	  \mu_{m'+1} \leq \alpha_2  < \mu_{m'} \quad:\quad m' \in \{1,\dotsc,d-1\}, \nonumber
	\end{gather}
	and satisfy
	\begin{gather}
	n = \min \{ n'\in \{1,\dotsc,d-r\}: \alpha_1(n') \leq \mu_{d-n'} \}, \label{eq:n1_as_min}\\
	m = \min \{ m'\in \{1,\dotsc,r\}: \alpha_2(m') \geq \mu_{m'+1} \}, \nonumber
	\end{gather}
	where $r$ is defined by 
	\begin{align}\label{def-r}
	\mu_{r+1} < \frac{1}{d}\leq \mu_{r}.
	\end{align}
	\end{lemma}

	Finally, if $\eps = 0$, set $\rhomaxeps(\sigma) = \sigma$. For $\eps>0$, define
	\begin{equation} \label{eq:def_rho_eps_sigma}
	\rho_\eps^*(\sigma) := \begin{cases}
	  \sum_{i\in I_L} \alpha_1 \ketbra{i}{i} + \sum_{i\in I_M} \mu_i \ketbra{i}{i} + \sum_{i\in I_H}\alpha_2 \ketbra{i}{i} & \text{ if }\frac{1}{2}\|\sigma-\tau\|_1 > \eps\\
	  \tau & \text{ if } \frac{1}{2}\|\sigma-\tau\|_1 \leq \eps
	\end{cases}
	\end{equation}
	where $(\alpha_1,n)$  and $(\alpha_2,m)$ are defined by Lemma~\ref{lem:char_of_alpha1}. For the case $\frac{1}{2}\|\sigma-\tau\|_1 > \eps$, from the construction of the state $\rho_\eps^*(\sigma)$ we can verify that its spectrum $\spec \rho$ lies in the interval $(0,1]$ and
	\[
	\tr [\rho_\eps^*(\sigma)] = \frac{n}{n}\left(\sum_{j\in I_L} \mu_j - \eps\right) + \sum_{j\in I_M} \mu_j + \frac{m}{m}\left(\sum_{j\in I_H}\mu_j - \eps\right) = \sum_{j=1}^d \mu_j=1,
	\]
	as well as
	\begin{equation}
	\|\rho_\eps^*(\sigma)-\sigma\|_1 = \sum_{j\in I_L}|\alpha_1-\mu_j| + \sum_{j\in I_H}|\alpha_2 - \mu_j| = \eps+\eps = 2 \eps; \label{eq:td_rho-eps-sigma}
	\end{equation}
	so for any $\epsilon\in[0,1]$, we have $\rho_\eps^*(\sigma)\in \Bep(\sigma)$. We briefly summarize two properties of $\rho_\eps^*(\sigma)$.
	\begin{proposition} \label{prop:VN_semigroup_property}
	The maximizer $\rho_\eps^*(\sigma)$ satisfies a semigroup property: if $\eps_1,\eps_2\in (0,1]$, we have
	\[
	\rho^*_{\eps_1 + \eps_2}(\sigma) = \rho^*_{\eps_1} ( \rho^*_{\eps_2}(\sigma)).
	\]
	\end{proposition}

	\begin{proposition} \label{prop:VN_opt_saturates_triangle} The state $\rho_\eps^*(\sigma)$ saturates the triangle inequality for the completely mixed state $\tau := \frac{1}{d}\one$ and $\sigma$, in that
	\[
	 \frac{1}{2}\|\sigma-\tau\|_1 =\frac{1}{2}\| \tau -  \rho_\eps^*(\sigma) \|_1+\frac{1}{2}\| \rho_\eps^*(\sigma) - \sigma\|_1  .
	\]
	\end{proposition}
	These results are proven in \Cref{sec:proofs_VN}. 

	In the next section, we prove that $\rhomaxeps(\sigma)$ satisfies the majorization order \eqref{eq:min-max-maj-cond}.
	\subsection{Minimality of $\rho^*_\eps(\sigma)$ in the  majorization order \eqref{eq:min-max-maj-cond}}\label{sec:max_of_pmaxeps}
Given $q = \vec \lambda(\sigma)$, we consider the vector $p_\epsilon^*(q)$, defined by 
\begin{equation}
(p_\eps^*(q))_j = \begin{cases}
\alpha_2 & j\in I_H \\
q_j & j \in I_M\\
\alpha_1 & j\in I_L,
\end{cases}
\end{equation}
which are the eigenvalues of $\rho^*_\eps(\sigma)$, defined in \eqref{eq:def_rho_eps_sigma}, and where $I_H$, $I_M$, and $I_L$ are defined by \eqref{eq:def_IH_IM_IL}. Let $p \in \Be(q)$, and consider its entries arranged in non-increasing order, 
\[
p_1\geq \dotsm \geq p_d.
\]
Let $q_1\geq \dotsm \geq q_d$ be the entries of $q$ in non-increasing order, and $p^*_1\geq \dotsm \geq p^*_d$ be the entries of $p_{\epsilon}^*(q)$ in non-increasing order. In this section, we show $p_{\epsilon}^* \prec p$. 

\begin{enumerate}
	\item First, we establish that $p^*_1 \leq p_1$.

To prove this, let us assume the contrary: $p^*_1 > p_1$. 
Then, since $p^*_1=p^*_2=\dotsm = p^*_{m} = \alpha_2(m)$,
\begin{align*}	
m\alpha_2(m)  = \sum_{i=1}^{m}p^*_i > m p_1 \geq \sum_{i=1}^{m}p_i.
\end{align*}

We conclude this step with the following lemma.
\begin{lemma} \label{lem:n2sum_contradiction}
If $m\alpha_2(m) > \sum_{i=1}^{m}p_i$, then $p\not\in \Be(q)$.
\end{lemma}
\begin{proof}	
Multiplying each side by $-1$ and adding $\sum_{i=1}^{m} q_i$, we have
\begin{equation}
\sum_{i=1}^{m}(q_i - p^*_i) < \sum_{i=1}^{m} (q_i - p_i) \leq \sum_{i=1}^{m} (q_i - p_i)_+ \leq \sum_{i=1}^{d} (q_i - p_i)_+ = \frac{1}{2}\|q - p\|_1. \label{eq:LB_norm_dist}
\end{equation}
Using $\alpha_2(m) = \frac{1}{m}\left(\sum_{i=1}^{m} q_i - \epsilon\right)$, the far left-hand side is
\[
\sum_{i=1}^{m}q_i - m\alpha_2(m)= \sum_{i=1}^{m}q_i  -\left(\sum_{i=1}^{m}q_i  - \epsilon\right) =\epsilon.
\]
Then \eqref{eq:LB_norm_dist} becomes
\[
\epsilon <  \frac{1}{2}\|q - p\|_1,
\]
contradicting that $p\in \Be(q)$.
\end{proof}

\item Next, for $k\in \{1,\dotsc,m\}$, we have $\sum_{i=1}^k p_i^* \leq \sum_{i=1}^k p_i$.

We prove this recursively: assume the property holds for some $k\in \{1,\dotsc,m-1\}$ but not for $k+1$. Note we have proven the base case of $k=1$ in the previous step. Then

\begin{equation}
\sum_{i=1}^{k} p_i^* \leq \sum_{i=1}^{k}p_i, \quad \text{and} \quad \sum_{i=1}^{k+1}p_i^* > \sum_{i=1}^{k+1} p_i. \label{eq:n2-assump}
\end{equation}
Subtracting the first inequality in (\ref{eq:n2-assump}) from the second, we have
\[
p_{k+1}^* > p_{k+1}
\]
yielding $\alpha_2(m)> p_{k+1} \geq p_{k} \geq \dotsm \geq p_{m}$. Summing the inequalities $p_{k+\ell}^*=\alpha_2(m) > p_{k+\ell}$ for $\ell=2,3,\dotsc,m-k$, we have
\[
\sum_{j=k+2}^{m} p_j^* > \sum_{j=k+2}^{m} p_j.
\]
Adding this to the second inequality of \eqref{eq:n2-assump}, we have
\[
\sum_{i=1}^{m}p_i^* > \sum_{i=1}^{m} p_i.
\]

We thus conclude by \Cref{lem:n2sum_contradiction}.

\item Next, let $k\in\{m+1,\dotsc,d-n\}$. Assume $\sum_{i=1}^k p_i^* > \sum_{i=1}^k p_i$. 
Then
\[
\sum_{i=1}^k q_i - p_i^* < \sum_{i=1}^k q_i - p_i
\]
However, the left-hand side is
\[
\sum_{i=1}^k q_i - p_i^* = \sum_{i=1}^{m} q_i - m \alpha_2(m) = \epsilon.
\]
Hence,
\[
 \epsilon < \sum_{i=1}^k q_i - p_i \leq \sum_{i=1}^d (q_i-p_i)_+ = \frac{1}{2}\|p-q\|_1
\]
which is a contradiction.

\item Finally, we finish with a recursive proof similar to step 2. Assume the property holds for  $k \in \{d-n,\dotsc,d-1\}$, but not for $k+1$. For this case too we have proven the base case $k=d-n$ in the previous step. We therefore assume
\begin{equation}
\sum_{i=1}^{k} p_i^* \leq \sum_{i=1}^{k} p_i, \qquad \text{and} \qquad \sum_{i=1}^{k+1} p_i^* > \sum_{i=1}^{k+1} p_i. \label{eq:contradiction>n1}
\end{equation}
Subtracting the two equations, we have
\[
p_{k+1}^* > p_{k+1}.
\]
Since $\alpha_1(n) = p_{k+1}^*$ we have  $\alpha_1(n) > p_{k+1}\geq  p_{k+1} \geq \dotsm \geq p_d$. Summing $p_{k+\ell}^* = \alpha_1(n) > p_{k+\ell}$ for $\ell=2,3,\dotsm,d$, we have
\[
 \sum_{i=k+2}^d p_i^* > \sum_{i=k+2}^d p_i.
\]
Adding to the second inequality of \eqref{eq:contradiction>n1}, we find
\[
1 = \sum_{i=1}^d p_i^* > \sum_{i=1}^d p_i
\]
which contradicts the assumption that $p\in \Delta$.\hfill\proofSymbol
\end{enumerate}

Thus, given $\sigma\in \cD$ and $\epsilon\in (0,1]$,  the state $\rho^*_\eps(\sigma)$ defined via \eqref{eq:def_rho_eps_sigma} has $\rho^*_\eps(\sigma)\prec \omega$ for any $\omega\in \Be(\sigma)$, proving the first majorization relation of \Cref{thm:max-min-ball}. Let us check the uniqueness of $\rho^*_\eps(\sigma)$. Assume another state $\tilde \rho \in \Be(\sigma)$ has
\[
\tilde \rho \prec \omega \qquad \forall\, \omega\in \Be(\sigma).
\]
Therefore, $\tilde \rho \prec \rho^*_\eps(\sigma) \prec \tilde \rho$, so $\tilde\rho$ and $\rho^*_\eps(\sigma)$ must be unitarily equivalent.
Moreover $\hat \rho :=\frac{1}{2}(\tilde \rho + \rho^*_\eps(\sigma)) \in \Be(\sigma)$ by convexity of the $\eps$-ball. By the strict concavity of the von Neumann entropy, $S(\hat \rho) > \frac{1}{2}(S(\tilde \rho) + S(\rho^*_\eps(\sigma))) = S(\rho^*_\eps(\sigma))$ using the unitary invariance of $S$. This contradicts that $\rho^*_\eps(\sigma)\prec \hat \rho$. Note that \eqref{eq:td_rho-eps-sigma} establishes the statement that $\rho_\eps(\sigma)$ lies on the boundary of $\Be(\sigma)$ when $\rho_\eps^*(\sigma)\neq \tau$. The remaining properties of $\rho_\eps^*(\sigma)$ stated in \Cref{thm:max-min-ball} are \Cref{prop:VN_semigroup_property} and \Cref{prop:VN_opt_saturates_triangle}, which are proved in \Cref{sec:proofs_VN}.

\subsection{Constructing the maximal state in the majorization order \eqref{eq:min-max-maj-cond}\label{sec:const_max_state_maj_order}}
	
As mentioned earlier, the minimum principle tells us that the entropy is minimized on a vertex of $\Be(q)$. In the example of \Cref{fig:Mw} in $d=3$, with $q=  ( 0.21, 0.24,0.55)$ and $\epsilon = 0.1$, we considered the set $M_w$ of points of $\Delta$ majorized by $w\in \Be(q)$. In the same example in \Cref{fig:min-max-example}, the minimum of the entropy over $\Be(q)$ was found to be $p_* =( 0.21 - \epsilon, 0.24,0.55+ \epsilon )$. In \Cref{fig:simplex_Mqmin_Mp}, we see that in fact $\Be(q) \subseteq M_{p_*}$. That is, $p \prec p_*$ for any $p\in \Be(q)$. In this section, we provide a general construction of $p_*$, and show that this property holds. As in the example, the construction will proceed by forming $p_*$ by decreasing the smallest entries of $q$ and increasing the largest entry. Intuitively, one ``spreads out'' the entries of $q$ to form $p_*$ so that $M_{p_*}$ covers the most area, in order to cover $\Be(q)$.

	\begin{figure}[ht]
	\centering
	\includegraphics{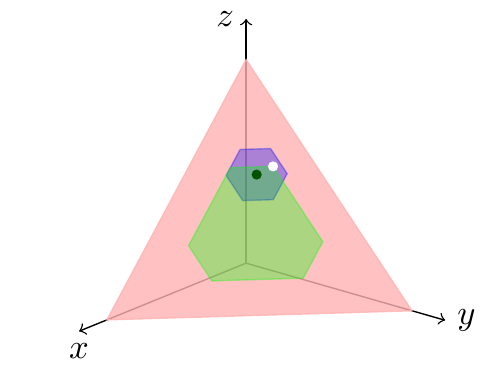}\qquad \includegraphics{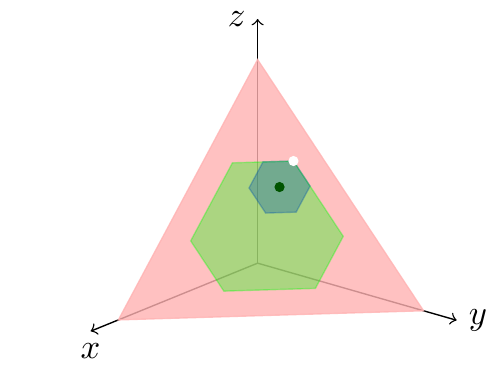}
	\caption{Left: The set $M_w$ for the point $w\in \Be(q)$ shown in white. Right: The set $M_{p_*}$ for $p_*$ the minimizer of $H$ over $\Be(q)$.\label{fig:simplex_Mqmin_Mp}}
	\end{figure}

	Let $\eps\in (0,1]$. We construct a probability vector $p_{*,\eps}(q)$ which we show has $M_{p_{*,\eps}(q)} \supseteq \Be(q)$ by using the definition of majorization given in \eqref{def:majorize}.
	\begin{definition}[$p_{*,\eps}(q)$]
	\label{def:pstarepsq}
	If $q_d > \epsilon$, let $p_{*,\eps}(q) = (q_1+\epsilon,q_2,\dotsc,q_d- \epsilon)$. Otherwise, let $\ell \in \{1,\dotsc,d-1\}$ be the largest index such that the sum of the $\ell$ smallest entries $Q_\ell := \sum_{j=d-\ell+1}^d q_j$ has $Q_\ell \leq \eps$. If $\ell=d-1$, set $p_{*,\eps}(q) = (1,0,\dotsc,0)$.
	 Otherwise, choose $p_{*,\eps}(q) = (p_{j*})_{j=1}^d$ for
	\begin{equation} \label{eq:def_pjstar}
	p_{j*} := \begin{cases}
	q_1 + \epsilon & j=1 \\
	q_j & 2\leq j \leq d-\ell-1\\
	q_{d-\ell+1} - (\epsilon - Q_\ell) & j=d-\ell\\
	0 & j\geq d-\ell+1.
	 \end{cases}
	\end{equation}
	\end{definition}
	We write $p_*$ for $p_{*,\eps}(q)$ in the following.
	Note that the condition $q_d > \epsilon$ is equivalent to every vector $p\in \Be(q)$ having strictly positive entries. This is the case of \Cref{fig:min-max-example}, and therefore $\pmineps(q)$ reduces to $(q_1+\epsilon,q_2,\dotsc,q_d- \epsilon)$.

	Using \Cref{def:pstarepsq}, we verify that $p_*\in \Be(q)$, as follows. First, if $q_d > \epsilon$, then
	\[
	\frac{1}{2}\|p_* - q\| = \frac{1}{2}(|q_1+\epsilon - q_1| + |q_d - \epsilon - q_d|) = \epsilon.
	\]
	 If $\ell = d-1$, then $p_* = (1,0,\dotsc,0)$ and
	\[
	 \frac{1}{2}\|p_* - q\| = \sum_{j=1}^d (q_j-p_{j*})_+ = \sum_{j=2}^{d} q_j = Q_\ell \leq \eps
	 \]yielding $p_*\in \Be(q)$. Otherwise, $\left(p_{j*}\right)_{j=1}^d$ is defined via \eqref{eq:def_pjstar}.
	For $j\neq d-\ell$, that $p_{j*} \geq 0$ is immediate, and by maximality of $\ell$, we have $p_{(d-\ell)*} = Q_{\ell+1} - \eps > 0$. Additionally, $\sum_{j=1}^d p_{j*} =   \sum_{j=1}^{d-\ell-1} q_j + Q_\ell -\eps + \eps = \sum_{j=1}^d q_j = 1$. Furthermore,
	\begin{align*}	
	\sum_{j=1}^d |p_{j*} - q_j| &= |q_1+\eps-q_1| +\sum_{j=2}^{d-\ell-1} |q_j - q_j| + |q_{d-\ell} - (\eps - Q_\ell) - q_{d-\ell}|  + \sum_{j=d-\ell+1}^d |q_j- 0| \\
	&= Q_\ell + (\eps - Q_\ell) + \eps = 2 \eps,
	\end{align*}
	so $p_* \in \Be(q)$.

	The following lemma shows that $p_*$ is indeed the maximal distribution in the majorization order \eqref{class-maj}.
	\begin{lemma}  \label{lem:pstar_succ_p}
	We have that $p_* \succ p$ for any $p\in \Be(q)$.
	\end{lemma}
	\begin{proof}	
	If $p_* = (1,0,\dotsc,0)$, then the result is immediate. If $p_* = (q_1+\epsilon,\dotsc,q_d - \epsilon)$, then consider $\ell=0$ and $Q_0 = 0$ in the following. Now, $p_* = (p_{j*})$ for $p_{j*}$ defined via \eqref{eq:def_pjstar}. Our task is to show that for any $p=\left(p_{j}\right)_{j=1}^d \in \Delta $,
	\begin{equation} \label{eq:proof_def_maj}
	\sum_{j=1}^k p_j \leq \sum_{j=1}^k p_{j*}
	\end{equation}
	holds for each $k\in\{1,\dotsc,d-1\}$. Equality in (\ref{eq:proof_def_maj}) obviously holds for $k=d$ since $p,p_* \in \Delta$.

	Since $\sum_{j=1}^d (p_j - q_j)_+ = \sum_{j=1}^d (p_j - q_j)_- \leq \eps$, in particular $\sum_{j=1}^k (p_j - q_j) \leq \eps$
	and therefore
	\[
	\sum_{j=1}^k p_j \leq \eps +\sum_{j=1}^k q_j.
	\]
	For $k \leq d-\ell-1$, we have $\sum_{j=1}^k p_{j*} = \eps +\sum_{j=1}^k q_j$, yielding \eqref{eq:proof_def_maj} in this case.
	On the other hand, for $k \geq d-\ell$ we have $\sum_{j=1}^k p_{j*} = \sum_{j=1}^d p_{j*} = 1$. Since $\sum_{j=1}^k p_j \leq 1$, this completes the proof.
	\end{proof}

	Given $\epsilon\in (0,1]$ and a quantum state $\sigma\in \cD$, with eigen-decomposition $\sigma =  \sum_{i=1}^d q_i\ketbra{i}{i}$ in the sorted eigenbasis for which $\sigma = \Eig^\downarrow(\sigma)$, we define
	\begin{equation} \label{eq:def_rho_eps_lowerstar}
	\rho_{\eps,*}(\sigma) = \sum_{i=1}^d p_{i*}(q) \ketbra{i}{i}
	\end{equation}
	where $p_{*,\eps}(q)$ is defined via \Cref{def:pstarepsq} and $q = \vec \lambda(\sigma)$. \Cref{lem:pstar_succ_p} therefore proves $\rho_{\eps,*}(\sigma) \succ \omega$ for any $\omega\in \Be(\sigma)$, proving the second majorization relation of \Cref{thm:max-min-ball}. The state $\rho_{\eps,*}(\sigma)$ is unique up to unitary equivalence as follows. If another state $\tilde \rho\in \Be(\sigma)$ had $\tilde \rho \succ \omega$ for all $\omega \in \Be(\sigma)$, then in particular, $\tilde \rho \succ \rho_{\eps,*}(\sigma) \succ \tilde \rho$, which implies that $\tilde \rho$ and $\rho_{\eps,*}(\sigma)$ are unitarily equivalent.
	
\subsection{Proofs for Section \ref{sec:const_min_state_maj_order}\label{sec:proofs_VN}}

	\paragraph{Proof of Lemma \ref{lem:char_of_alpha1}}
	Here we prove the results pertaining to the pair $(\alpha_2,m)$. The results for the pair $(\alpha_1,n)$ can be obtained analogously. Note that if any pair $(\alpha_2,m) \in [0,1]\times \{1,\dotsc,d-1\}$ satisfies \eqref{eq:def_alpha1_n1} then we have
		\begin{equation*} 
		\eps = \sum_{i=1}^{m}|\alpha_2 - \mu_i| =  \sum_{i=1}^{m}( \mu_i-\alpha_2 ) = \sum_{i=1}^{m}\mu_i-m \alpha_2 ,
		\end{equation*}
		implying $\alpha_2 = \frac{1}{m}\left( \sum_{i=1}^{m}\mu_i - \eps \right) = \alpha_2(m)$.
		Conversely, if for some $m'\in \{1,\dotsc,d-1\}$, the corresponding value $\alpha_2(m')$ satisfies $\mu_{m'+1} \leq \alpha_2(m') < \mu_{m'}$, then
		\[
		 \eps = \sum_{i=1}^{m'}\mu_i- m' \alpha_2(m') = \sum_{i=1}^{m'}( \mu_i-\alpha_2(m'))=\sum_{i=1}^{m'}|\alpha_2(m')- \mu_i| 
		\]
		and in particular
		\[
		\sum_{i=1}^{m'}|\alpha_2(m')- \mu_i| = \eps.
		\]
		Hence, the existence and uniqueness of $(\alpha_2,m)$ satisfying \eqref{eq:def_alpha1_n1} is equivalent to the existence and uniqueness of $m\in \{1,\dotsc,d-1\}$ such that $\alpha_2(m)$ satisfies $\mu_{m+1} \leq \alpha_2(m) < \mu_{m}$.

		Next, we show that
		\[
		m=\min \{m'\in \{1,\dotsc,r\}: \alpha_2(m')\geq \mu_{m'+1} \}
		\]
		by checking that the minimum exists and uniquely solves $\mu_{m+1} \leq \alpha_2(m) < \mu_{m}$. The proof is then completed by showing that we must have $\alpha_2(m) > \frac{1}{d}$.
	The steps of the construction are elucidated below.

	~\begin{enumerate}[label={Step }\arabic*.]
		
		\item $ \{ m'\in \{1,\dotsc,r\}:\alpha_2(m')\geq \mu_{m'+1}\}\neq \emptyset$.

	Let us assume the contrary. Then, in particular, that $\alpha_2(r) < \mu_{r+1}$. By substituting $\alpha_2(r) = \frac{1}{r}\sum_{i=1}^r \mu_i - \frac{1}{r}\eps$ in this inequality, we have 
	\begin{equation} \label{eq:lem_proofalpha1n1_contradiction1}
	\sum_{i=1}^r \mu_i< r \mu_{r+1} + \eps 
	\end{equation}
	Using $\tr(\sigma-\tau) = 0$, we have $\frac{1}{2}\|\sigma-\tau\|_1 = \tr( \tau-\sigma)_-$. Using the definition of $r$, \cref{def-r}, this can
	be written as
	\begin{equation}	
	\frac{1}{2}\|\sigma-\tau\|_1 = \sum_{i=1}^r (\mu_i-\frac{1}{d}) = \sum_{i=1}^r \mu_i-\frac{r}{d } \label{eq:tracedist_sigma-pi}.
	\end{equation}
	Employing \eqref{eq:lem_proofalpha1n1_contradiction1}, and using the definition of $r$, we find that
	\begin{align}
	\frac{1}{2}\|\sigma-\tau\|_1 <  r \mu_{r+1} + \eps -\frac{r}{d}  =  \eps + r ( \mu_{r+1} - \frac 1d) <  \eps.
	\end{align}
	That is, $\tau\in B_\eps(\sigma)$, which contradicts our assumption.

	\item The value $x := \min \{ m'\in [r]: \alpha_2(m') \geq \mu_{m'+1}\}$ solves $\mu_{x+1} \leq \alpha_2(x) <  \mu_{x}$.

	If $x=1$, then using \cref{eq-alpha} we see that $\mu_1 > \mu_1 - \eps = \alpha_2(1)$, and hence $\mu_{x+1} \leq \alpha_2(x) <  \mu_{x}$.
	Otherwise, by minimality of $x$, we have $\alpha_2(x-1) < \mu_x$.
	We first establish  that for $m'\in \{2,3,\dotsc,d\}$,
	\begin{equation} 
	\label{eq:alpha_id}
	\alpha_2(m'-1) < \mu_n \iff \alpha_2(m')  < \mu_n
	\end{equation}
	and the result follows by taking $m'=x$. To prove \eqref{eq:alpha_id}, we write
	\begin{align*}	
	\alpha_2(m'-1)< \mu_n & \iff\frac{1}{m'-1}[\sum_{j=1}^{m'-1}\mu_j - \eps]< \mu_n \\
	&\iff \sum_{j=1}^{m'-1}\mu_j - \eps< (m'-1)\mu_n = m' \mu_n - \mu_n \\
	&\iff \sum_{j=1}^{m'}\mu_j - \eps<  m' \mu_n\\
	&\iff \frac{1}{m'}[\sum_{j=1}^{m'}\mu_j - \eps]<  \mu_n\\
	&\iff \alpha_2(m') < \mu_n.
	\end{align*}

	\item Uniqueness of $m$ satisfying $\mu_{m+1} \leq \alpha_2(m)< \mu_{m}$.

	By substituting the inequality $\mu_{m'-1} \geq \mu_n$ into \eqref{eq:alpha_id}, we find for $m'\in \{ 2,3,\dotsc,d\}$,
	\begin{equation} 
	\label{eq:alpha_id2}
	 \mu_n > \alpha_2(m')  \implies  \mu_{m'-1} > \alpha_2(m'-1).
	\end{equation}
	Now, assume that there exists a $y$ satisfying $y>x>0$ for which $\mu_y > \alpha_2(y) \geq \mu_{y+1}$. By applying the implication \eqref{eq:alpha_id2} a total of $(y-x-1)$ times, we see that
	\[
	\mu_y > \alpha_2(y) \implies \mu_{y-1} > \alpha_2(y-1) \implies \dotsm \implies \mu_{x+1} > \alpha_2(x+1)
	\]
	which by \eqref{eq:alpha_id} is equivalent to $\alpha_2(x) < \mu_{x+1}$.
	This contradicts the assumption that $\alpha_2(x) \geq \mu_{x+1}$. Hence, such a $y$ cannot exist.

	\item $\alpha_2(m) > \frac{1}{d}$.

	We prove this by showing that if $\alpha_2(m) \leq \frac{1}{d}$ then we obtain a contradiction to the assumption $\frac{1}{2}\|\sigma-\tau\|_1 >  \eps$
	of \Cref{lem:char_of_alpha1}.

	Assume  $\alpha_2(m) \leq \frac{1}{d}$. Then,
	\begin{align*}
	\alpha_2(m) = \frac{1}{m}\sum_{j=1}^{m} \mu_j - \frac{\eps}{m}	 \leq \frac{1}{d}  \iff
	\sum_{j=1}^{m} \left(\mu_j- \frac{1}{d} \right) \leq \eps.
	\end{align*}
	Now, since $m\leq r$ by \eqref{eq:n1_as_min}, $\mu_j -\frac{1}{d}\geq 0$ for each $j\in [m]$. 

	If $m=r$, then 
	\begin{equation}
	\sum_{j=1}^{r}  \left(\mu_j- \frac{1}{d} \right)\leq \eps. \label{eq:pos_part_leq_eps}
	\end{equation}
	On the other hand, if $m <r$, then $m+1 \leq r$. Then, using the assumption $\alpha_2(m)\leq \frac{1}{d}$,
	\[
	\frac{1}{d} \geq \alpha(m) \geq \mu_{m+1} \geq \dotsm \geq \mu_r \geq \frac{1}{d},
	\]
	so $\mu_{m+1} = \mu_{m+2} =\dotsm = \mu_r = \frac{1}{d}$. Then,
	\eqref{eq:pos_part_leq_eps}
	holds in this case as well. Then by \eqref{eq:tracedist_sigma-pi},
	\[
	\frac{1}{2}\|\sigma - \tau\|_1 =  \sum_{j=1}^r \left(\mu_j- \frac{1}{d} \right)\leq \eps,
	\]
	which contradicts the assumption that $\tau \not\in B_\eps(\sigma)$.\hfill\proofSymbol
	\end{enumerate}

	\paragraph{Proof of Proposition \ref{prop:VN_semigroup_property}} We first treat the cases in which $\rho_{\eps_1+ \eps_2}(\sigma) = \tau$; that is, when $\frac{1}{2}\|\sigma-\tau\|_1 \leq \eps_1+\eps_2$.
	If $\frac{1}{2}\|\sigma-\tau\|_1 \leq \eps_2$, then  $ \rho_{\eps_2}(\sigma) = \tau$ as well; since $\rho_{\eps_1}(\tau)=\tau$, we have $\rho_{\eps_1 + \eps_2}(\sigma) = \rho_{\eps_1}(\rho_{\eps_2}(\sigma))$. Next, consider the case  $\eps_2 < \frac{1}{2}\|\sigma-\tau\|_1 \leq \eps_1 +\eps_2$.  To show $ \rho_{\eps_1} ( \rho_{\eps_2}(\sigma)) = \tau$, we need $\frac{1}{2}\|\rho_{\eps_2}(\sigma) - \tau\|_1 \leq \eps_1$.
	By Proposition \ref{prop:VN_opt_saturates_triangle},
	\[
	 \frac{1}{2}\|\tau - \rho_{\eps_2}(\sigma)\|_1  =   \frac{1}{2}\|\sigma-\tau\|_1 - \frac{1}{2}\| \rho_{\eps_2}(\sigma) - \sigma\|_1  \leq \eps_1 + \eps_2 -  \frac{1}{2}\| \rho_{\eps_2}(\sigma) - \sigma\|_1
	\]
	but since $ \frac{1}{2}\| \rho_{\eps_2}(\sigma) - \sigma\|_1 = \eps_2$, using that $\rho_{\eps_2}(\sigma)\neq \tau$, we have $ \frac{1}{2}\|\tau - \rho_{\eps_2}(\sigma)\|_1  \leq \eps_1$ as required. This completes the proof of the cases for which $\rho_{\eps_1+\eps_2}(\sigma)=\tau$.

	Now, assume that $\frac{1}{2}\|\sigma-\tau\|_1 > \eps_1 + \eps_2$. Let $\sigma = \sum_{i=1}^d \mu_i \ketbra{i}{i}$ and write
	\[
	\rho_{\eps_2}(\sigma) = \sum_{i\in I_L} \alpha_1 \ketbra{i}{i} + \sum_{i\in I_M} \mu_i \ketbra{i}{i} + \sum_{i\in I_H}\alpha_2 \ketbra{i}{i}
	\]
	for $I_L = \{d-n+1,\dotsc,d\}$, $I_M = \{m+1,\dotsc,d-n\}$, and $I_H = \{1,\dotsc,m\}$, and  $(\alpha_1,n)$ and $(\alpha_2,m)$ are determined by $\sigma$ and $\eps_2$ via \Cref{lem:char_of_alpha1}.

	Now, consider
	\begin{align}\label{rho-1-2}
	\rho_{\eps_1}( \rho_{\eps_2}(\sigma) ) =\sum_{i\in I_{L'}} \beta_1 \ketbra{i}{i} + \sum_{i\in I_{M'}} \mu_i \ketbra{i}{i} + \sum_{i\in I_{H'}} \beta_2 \ketbra{i}{i} 
	\end{align}
	with $I_{L'} = \{d-n'+1,\dotsc,d\}$, $I_{M'} = \{m'+1,\dotsc,d-n'\}$, and $I_{H'} = \{1,\dotsc,m'\}$, and where $(\beta_1,n')$ and $(\beta_2,m)$ are determined by $ \rho_{\eps_2}(\sigma)$ and $\eps_1$ via \Cref{lem:char_of_alpha1}.

	We aim to compare the expression (\ref{rho-1-2}) of $\rho_{\eps_1}( \rho_{\eps_2}(\sigma) ) $ to the following:
	\[
	\rho_{\eps_1 + \eps_2} (\sigma) = \sum_{i\in I_{L''}} \gamma_1 \ketbra{i}{i} + \sum_{i\in I_{M''}}\mu_i \ketbra{i}{i} + \sum_{i\in I_{H''}} \gamma_2 \ketbra{i}{i}
	\]
	with $I_{L''} = \{d-n''+1,\dotsc,d\}$, $I_{M''} = \{m''+1,\dotsc,d-n''\}$, and $I_{H''} = \{1,\dotsc,m''\}$, and where $(\gamma_1,n'')$ and $(\gamma_2,m'')$ are determined by $(\sigma,\eps_1+\eps_2)$ via \Cref{lem:char_of_alpha1}. That is, we wish to show $(\beta_1,n') = (\gamma_1,n'')$, and $(\beta_2,m') = (\gamma_2,m'')$. We only consider the first equality here; the second is very similar.

	 Let $(\nu)_{i=1}^d$ be the eigenvalues of  $\rho_{\eps_2}(\sigma)$. That is, $\nu_i = \alpha_1$ for $i\in I_L$, $\nu_i = \mu_i$ for $i \in I_M$, and $\nu_i = \alpha_2$ for $i\in I_H$. Equation~\eqref{eq:n1_as_alpha1(n1)_inequality} for $\rho_{\eps_1}(\rho_{\eps_2}(\sigma))$ in Lemma~\ref{lem:char_of_alpha1} yields
	\begin{equation}
	\nu_{m'} > \beta_2 = \frac{1}{m'}\left(\sum_{i\in I_{H'}} \nu_i - \eps_1\right) \geq  \nu_{m'+1}. \label{eq:defd_beta1_(proof)}
	\end{equation}
	Thus,
	\begin{itemize}
		\item $m' \geq m$. Otherwise, $\beta_2 = \alpha_2 - \frac{\eps_1}{m'} < \alpha_2 = \nu_{m'+1}$.
		\item $m' \leq d- n$. Otherwise, $\beta_2 < \nu_{m'} = \alpha_1 < \frac{1}{d}$, contradicting that $\beta_2 > \frac{1}{d}$ by Lemma~\ref{lem:char_of_alpha1}.
	\end{itemize}
	Hence,
	\begin{align*}	
	\beta_2 &= \frac{1}{m'}\left(\sum_{i\in I_H} \nu_i + \sum_{m+1}^{m'} \nu_i - \eps_1 \right) =\frac{1}{m'}\left(m\alpha_2 + \sum_{m+1}^{m'} \mu_i - \eps_1 \right)\\
	&=\frac{1}{m'}\left(\sum_{i\in I_H}\mu_i - \eps_2 + \sum_{m+1}^{m'} \mu_i - \eps_1 \right) = \frac{1}{m'} \left( \sum_{i\in I_{H'}} \mu_i - \eps_1 -\eps_2 \right).
	\end{align*}
	That is, $\beta_2 = \gamma_2(m')$. It remains to show $m' = m''$.

	\begin{itemize}
	\item If $m'=m$, then $\nu_{m'} = \alpha_2 < \mu_{m'}$ by equation \eqref{eq:def_alpha2_m} for $\rho_{\eps_2}(\sigma)$.
	\item If $m' > m$, then $ \nu_{m'} =\mu_{m'}$. 
	\end{itemize}

In either case, $\nu_{m'} \leq \mu_{m'}$. Hence, \eqref{eq:defd_beta1_(proof)} becomes
	\[
	\mu_{m} \geq \nu_{m'}  > \beta_2 \geq \nu_{m+1}.
	\]
	Either $\nu_{m+1} = \mu_{m+1}$, or $\nu_{m+1} = \alpha_1 > \mu_{m+1}$; in either case, $\beta_2 \geq \mu_{m+1}$. Thus, writing $\beta_2 = \gamma_2(m')$, we have
	\[
	\mu_{m'} > \gamma_2(m') \geq \mu_{m'+1}.
	\]
	Equation~\eqref{eq:n1_as_alpha1(n1)_inequality} in \Cref{lem:char_of_alpha1} defines $m''$ as the unique solution of 
	\[
	 \mu_{m''} >  \gamma_2(m'') \geq \mu_{m''+1} 
	\]
	and therefore $m'' = m'$. \hfill\proofSymbol

\paragraph{Proof of Proposition \ref{prop:VN_opt_saturates_triangle}}
	Note \[
	\frac{1}{2}\| \rho_\eps^*(\sigma) - \sigma\|_1 = \min\left(\eps, \frac{1}{2}\|\sigma-\tau\|_1\right),
	\]
	 from the construction of $\rho_\eps^*(\sigma)$.
	Now, if  $\rho_\eps^*(\sigma)=\tau$, we have the result immediately. Otherwise, 
	\begin{align*}  
	\|\tau -\rho_{ \eps}(\sigma)\|_1 
	&=  n\left( \frac{1}{d} - \alpha_1 \right) +   m\left( \alpha_2-\frac{1}{d}  \right) + \sum_{j\in I_M} | \tfrac{1}{d}-\mu_j|\\
	&= \frac{n}{d} - \sum_{j\in I_L} \mu_j -  \eps +\sum_{j\in I_H} \mu_j - \frac{m}{d} -  \eps + \sum_{j\in I_M} | \tfrac{1}{d}-\mu_j|\\
	&= \sum_{j=1}^{d} | \tfrac{1}{d}-\mu_j| - 2 \eps\\
	&= \| \tau - \sigma\|_1 - 2  \eps = \| \tau - \sigma\|_1 - \|\rho_\eps^*(\sigma)- \sigma\|_1 . \tag*{\proofSymbol}
	\end{align*}

\section{Proof of the local continuity bound (\Cref{prop:local_cont_bound})}

Our local continuity bound for the von Neumann entropy, given by the inequality (\ref{eq:local_cont_bound}) of \Cref{prop:local_cont_bound}, is
an immediate consequence of \Cref{cor:main_thm_for_CE}, which in turn follows directly from \Cref{thm:main_convex_result}. This can be seen by noting that 
\[
	 S(\rho_{*,\eps}) \leq S(\omega)  \leq S(\rho_\eps^*(\sigma)) ,
	\]
	by Schur concavity of the von Neumann entropy, and the minimality of $\rho_\eps^*(\sigma)$ and maximality of $\rho_{*,\eps}$ in the majorization order (\ref{eq:min-max-maj-cond}). 

The inequality (\ref{eq:AF_from_local}) follows by applying the (AF)-bound (\Cref{lem:AF_equality}) to the pairs of states $\left(\rho_\eps^*(\sigma), \sigma\right)$ and $\left(\sigma, \rho_{*,\eps}\right)$ as follows.
The (AF)-bound gives
	\begin{equation}
	 S(\rho_\eps^*(\sigma)) - S(\sigma) \leq \AF(\epsilon), \qquad  S(\sigma) - S(\rho_{*,\eps}(\sigma)) \leq \AF(\epsilon) \label{eq:AF_bounds_each}
	\end{equation}
	for 
	\begin{equation}
		 \AF(\epsilon) := \begin{cases}
 \epsilon \log (d-1) + h(\epsilon) & \text{if } \epsilon < 1 - \tfrac{1}{d} \\
	 \log d & \text{if } \epsilon \geq 1 - \tfrac{1}{d}.
 \end{cases}  \label{def:AFeps}
	\end{equation}
What remains to be established is the necessary and sufficient condition for equality in (\ref{eq:AF_from_local}). A sufficient condition for equality in the (AF)-bound was obtained by Audenaert, and is stated 
in \Cref{lem:AF_equality}.
Here we prove that this condition is also necessary. In order to do this, it is helpful to recall the proof of the (AF)-bound in detail.

\paragraph{Proof of \Cref{lem:AF_equality}}
Without loss of generality, assume $S(\rho) \geq S(\sigma)$.

Let us first consider $\epsilon\geq 1 - \frac{1}{d}$. As $0\leq S(\omega) \leq \log d$ for any state $\omega\in \cD$, the left-hand side of \eqref{eq:Audenaert-Fannes_bound} is bounded by $S(\rho) - S(\sigma) \leq \log d - 0= \log d$, which is the right-hand side of the inequality \eqref{eq:Audenaert-Fannes_bound} in this range of values of $\eps$.  Moreover, equality is achieved if and only if $S(\rho) = \log d$ and $S(\sigma)=0$. As only the completely mixed state, $\tau = \frac{1}{d}\one$, achieves $S(\tau) = \log d$, and only pure states have zero entropy, we have established the equality conditions for $\epsilon \geq 1- \frac{1}{d}$. Note that $\tau \in B_\eps(\sigma)$ for any state $\sigma \in \D$ and $\epsilon\geq 1 - \frac{1}{d}$.
\smallskip

Next, let $\epsilon < 1 - \frac{1}{d}$. Note that if $\rho = \diag(1- \epsilon, \frac{\epsilon}{d-1},\dotsc, \frac{\epsilon}{d-1})$, then $S(\rho) = h(\epsilon) + \epsilon \log (d-1)$. If $\sigma$ is a pure state, then $S(\sigma)=0$. Hence, these choices of $\rho$ and $\sigma$ are sufficient to attain equality in \eqref{eq:Audenaert-Fannes_bound}. It remains to show that in this range of values of $\eps$, equality is achieved in the (AF)-bound only for this pair of states. To show this, we first reduce the problem to a classical one. This is because, by comparing the $1$-norm on $\R^d$ to the trace norm between mutually diagonal matrices, we have
\begin{equation}
\|\vec \lambda(\rho) - \vec \lambda(\sigma)\|_1 = \| \Eig^\downarrow(\rho) - \Eig^\downarrow(\sigma)\|_1 \leq \|\rho-\sigma\|_1 \leq \epsilon \label{eq:AF_lambda-rho_lambda-sigma}
\end{equation}
where the first inequality follows from  \Cref{lem:eig_upperbounded_trace-dist}.
Since $H(\vec \lambda (\rho)) = S(\Eig^\downarrow(\rho)) = S(\rho)$, and $H(\vec \lambda(\sigma)) = S(\Eig^\downarrow(\sigma)) = S(\sigma)$, we can reduce to the classical case of probability distributions on $p,q\in\Delta$. This argument is due to \cite{Audenaert07}. Thus the (AF)-bound is equivalently stated in terms of probability distributions as in the following lemma.

\begin{lemma} \label{lem:AF_classical_H}
Let $\epsilon\in (0,1]$ and $p,q\in \Delta$ such that $\frac{1}{2}\|p-q\|_1 \leq \epsilon$, then
\begin{equation}
|H(p) - H(q)| \leq \begin{cases}
 \epsilon \log (d-1) + h(\epsilon) & \text{if } \epsilon < 1 - \tfrac{1}{d} \\
	 \log d & \text{if } \epsilon \geq 1 - \tfrac{1}{d}.
 \end{cases}  \label{eq:AF_classical_H}
\end{equation}
Without loss of generality, assume $H(p) \geq H(q)$. Then equality in \eqref{eq:AF_classical_H} occurs if and only if for some $\pi \in S_d$, we have $\pi(q) = (1,0,\dotsc,0)$, and either
\begin{enumerate}
	\item  $\epsilon < 1 - \frac{1}{d}$,  and $\pi(p) = ( 1- \epsilon, \frac{\epsilon}{d-1},\dotsc, \frac{\epsilon}{d-1})$, or
	\item $\epsilon \geq 1- \frac{1}{d}$, and $p = ( \frac{1}{d},\dotsc, \frac{1}{d})$.
\end{enumerate}
\end{lemma}
Using this result, by setting $p:= \vec \lambda (\rho)$, and $q:= \vec \lambda(\sigma)$, by \eqref{eq:AF_lambda-rho_lambda-sigma} and \Cref{lem:AF_classical_H}, we obtain inequality~\eqref{eq:Audenaert-Fannes_bound}. Moreover, to attain equality in \eqref{eq:Audenaert-Fannes_bound}, we require equality in \eqref{eq:AF_classical_H}. This fixes $\rho = \diag (1 -\epsilon, \frac{\epsilon}{d-1},\dotsc, \frac{\epsilon}{d-1})$ in some basis, and $\sigma$ a pure state.  \hfill\proofSymbol

\medskip
\noindent It remains to prove the equality condition in \Cref{lem:AF_classical_H} for the range $\epsilon < 1 - \frac{1}{d}$.
\paragraph{Proof of \Cref{lem:AF_classical_H}} To establish the inequality \eqref{eq:AF_classical_H}, we recall the proof presented in \cite{Winter16} in detail\footnote{This proof is originally due to Csisz\'ar (see \cite[Theorem 3.8]{PetzQITbook}).}. This also helps us to investigate when equality occurs, and deduce the form of $p$ and $q$. 

Without loss of generality, assume $H(p) \geq H(q)$. A coupling of $(p,q)$ is a probability measure $\omega$ on $[d]\times [d]$, for $[d] = \{1,\dotsc,d\}$,  such that $\sum_{i\in [d]} \omega(i,j) = q(j)$ and $\sum_{j\in [d]} \omega(i,j) = p(i)$. For any coupling $\omega$, it is known that, if $(X,Y)$ are a pair of random variables with joint measure $\omega$, i.e.~$(X,Y)\sim \omega$, then $ \Pr[X\neq Y] \geq \frac{1}{2}\|p-q\|_1.$
Moreover, there exist optimal couplings $\omega^*$ which achieve equality: if $(X,Y)\sim \omega^*$, then $\Pr[X\neq Y] = \frac{1}{2}\|p-q\|_1$.
 In fact, we can choose 
 \begin{equation}
 \omega^*(i,j) := \begin{cases}
 \min(p(i),q(i)) & \text{if } i=j\\
 \frac{(p(i)-q(i))_+ \cdot (q(j) - p(j))_+}{\frac{1}{2}\|p-q\|_1} & \text{otherwise},
 \end{cases} \label{eq:def_omegastar}
 \end{equation}
 using the notation $(z)_+ = \max(z,0)$ for $z\in \RR$.

 \begin{lemma}
 $\omega^*$ is a maximal coupling of $(p,q)$.
 \end{lemma}
 \begin{proof}	
 First, $\sum_{i} \omega^*(i,i) = \sum_i \min(p(i),q(i)) = \frac{1}{2}\sum_i [p(i) + q(i) - |p(i) - q(i)|] = 1 - \frac{1}{2}\|p-q\|_1$, and therefore $\Pr[X\neq Y] = 1 - \Pr[X=Y] = \frac{1}{2}\|p-q\|_1$.
Moreover, $\omega^*$ is a coupling of $p$ and $q$. We have
 \begin{equation*}
  \sum_j w^*(i,j) =  \min(p(i),q(i)) +\frac{(p(i)-q(i))_+ }{\frac{1}{2}\|p-q\|_1} \sum_{j\neq i}(q(j) - p(j))_+. \label{eq:coupling_proof}
  \end{equation*} If $p(i) \leq q(i)$, then $(p(i)-q(i))_+=0$, and hence $ \sum_j w^*(i,j) = p(i)$. Otherwise,
 \[
 \sum_j w^*(i,j)  = q(i) + \frac{p(i) - q(i)}{\frac{1}{2}\|p-q\|_1} \sum_{j} (q(j) - p(j))_+.
 \]
 However, $\sum_{j} (q(j) - p(j))_+ = \frac{1}{2}\|p-q\|_1$, yielding $ \sum_j w^*(i,j) = p(i)$. Checking  $ \sum_i w^*(i,j)  = q(j)$ is similar.
 \end{proof}

Let $(X,Y)\sim \omega^*$. As $H(X,Y) \geq H(X)$, we have
 \begin{equation} \label{ieq:X}
 H(X) - H(Y) \leq H(X,Y) - H(Y) = H(X|Y) \leq h(\epsilon) + \epsilon \log(d-1)
 \end{equation}
 where the second inequality is Fano's inequality, \Cref{lem:Fano}. Since $H(X) = H(p)$ and $H(Y) = H(q)$, this concludes the proof of the inequality \eqref{eq:AF_classical_H}.

 Now, assume we attain equality in \eqref{eq:AF_classical_H}. Note first that we must have $\frac{1}{2}\|p-q\|_1 = \epsilon$. Otherwise, if $\frac{1}{2}\|p-q\|_1 = \epsilon' < \epsilon < 1 - \frac{1}{d}$, then by \eqref{eq:AF_classical_H},
 \[
  H(p) - H(q) \leq \epsilon' \log(d-1) + h(\epsilon') < \epsilon \log(d-1) + h(\epsilon)
  \] 
  by the strict monotonicity of $\eps\mapsto \eps\log(d-1) + h(\eps)$  for $\eps\in [0, 1-\frac{1}{d})$, which can be confirmed by differentiation. As $H(X,Y) = H(X) + H(Y|X)$, to have equality in $H(X,Y) \geq H(X)$, we require $H(Y|X)=0$. The second inequality in \eqref{ieq:X} is Fano's inequality, and to obtain equality, we require
\begin{equation} \label{eq:Fano_X}
\frac{\omega^*(i,j)}{q(j)} = \begin{cases}
1- \epsilon & i=j \\
\frac{\epsilon}{d-1} & i \neq j
\end{cases}
\end{equation}
whenever $q(j)\neq 0$, by \Cref{lem:Fano}. Let $k$
be such that $q(k) > 0$. Then by \eqref{eq:def_omegastar} and \eqref{eq:Fano_X},
\[
\omega^*(k,k) = \min(p(k), q(k)) = (1- \epsilon)q(k).
\]
Since $q(k)>0$, we cannot have $q(k) = (1 -\epsilon) q(k)$, and therefore $p(k)  = (1 -\epsilon) q(k)$. Next, by \eqref{eq:Fano_X}, for $i\neq k$,
\[
   q(k)\frac{\epsilon}{d-1} = \omega^*(i,k) = \frac{((p(i) - q(i))_+ \cdot (q(k) - p(k))_+}{\epsilon}
\]
using \eqref{eq:def_omegastar} for the second equality. Since $(q(k) - p(k))_+ = \epsilon q(k)$, we have $(p(i) - q(i))_+ = \frac{\epsilon}{d-1}$, and thus
\begin{equation} \label{eq:pi_qi_proof}
p(i) = q(i) + \frac{\epsilon}{d-1}.
\end{equation}
Now, if we sum \eqref{eq:pi_qi_proof} over all $i$ such that $i\neq k$, we obtain
\[
1 - p(k) = \sum_{i\neq k} p(i) = \sum_{i \neq k} q(i) + \epsilon = 1 - q(k) + \epsilon.
\]
Substituting $p(k) = (1 -\epsilon) q(k)$, we have
\[
 \epsilon q(k) = 1 -q(k) + \epsilon
\]
That is, $q(k) = 1$, and therefore $p(k) = 1- \epsilon$. This fixes $q(i) = 0$ for $i\neq k$, so \eqref{eq:pi_qi_proof} yields $p(i) = \frac{\epsilon}{d-1}$ for $i\neq k$. This completes the proof. 
\smallskip

For completeness, one can note that although we did not use the assumption $H(Y|X)=0$ directly, the deductions above from equality in Fano's inequality yield
\[
\omega^*(i,j) = \begin{cases}
0 & \text{if } i= j \neq k \text{ or } j\neq i = k,\\
\frac{\epsilon}{d-1} & \text{if } i\neq j = k, \\
1- \epsilon & \text{if } i\neq j = k
\end{cases} \implies  \frac{\omega^*(i,j)}{p(i)} = \begin{cases}
0 & \text{if }j\neq k\\
1 & \text{if }j =k,
\end{cases}
\]
and therefore $H(Y|X) =0$, as required.\hfill\proofSymbol

\paragraph{Proof of Proposition~\ref{prop:local_cont_bound}} From the discussion at the beginning of the Section, it only remains to establish the equality conditions for Equation~\eqref{eq:AF_from_local}. By \eqref{eq:AF_bounds_each}, if either $S(\rho^*_\epsilon(\sigma)) - S(\sigma) = \AF(\epsilon)$ or $S(\sigma) - S(\rho_{*,\eps}(\sigma)) = \AF(\epsilon)$, then we achieve equality in \eqref{eq:AF_from_local}. Here, $\AF(\eps)$ is as given by \eqref{def:AFeps}.

	If $\sigma$ is a pure state, in some basis we can write $\sigma = \ketbra{0}{0}$. Then
	\[
	\rho_\eps^*(\sigma) = (1- \eps)\ketbra{0}{0} + \sum_{i=1}^{d-1} \frac{\eps}{d-1} \ketbra{i}{i}
	\]
	yielding $S(\rho_\eps^*(\sigma)) = -(d-1) \frac{\eps}{d-1}\log(\frac{\eps}{d-1}) -(1- \eps)\log(1- \eps) =  \AF(\epsilon)$.
	Likewise, if $\sigma= (1- \eps)\ketbra{0}{0} + \sum_{i=1}^{d-1} \frac{\eps}{d-1} \ketbra{i}{i}$, then $\ketbra{0}{0} \in \Be(\sigma)$. Hence, $\rho_{*,\epsilon}(\sigma) = \ketbra{0}{0}$, and by the same computation, as for the pure state case, we recover $S(\sigma) - S(\rho_{*,\eps}(\sigma)) = \AF(\epsilon)$.

	On the other hand, let us assume equality in \eqref{eq:AF_from_local}. Then in particular, either $S(\rho^*_\epsilon(\sigma)) - S(\sigma) = \AF(\epsilon)$ or $S(\sigma) - S(\rho_{*,\eps}(\sigma)) = \AF(\epsilon)$.
	\begin{enumerate}
		\item Case 1: $S(\rho^*_\epsilon(\sigma)) - S(\sigma) = \AF(\epsilon)$. Since $S(\rho^*_\epsilon(\sigma)) \geq S(\sigma)$, we have by \Cref{lem:AF_equality} that $\sigma$ must be a pure state.
		\item Case 2: $S(\sigma) - S(\rho_{*,\eps}(\sigma)) = \AF(\epsilon)$. Then by \Cref{lem:AF_equality}, as $S(\sigma) \geq S(\rho_{*,\epsilon})$, if $\epsilon < 1 - \frac{1}{d}$, then $\sigma = \diag ( 1 - \epsilon, \frac{\epsilon}{d-1}, \dotsc, \frac{\epsilon}{d-1})$, and if $\epsilon\geq 1 - \frac{1}{d}$, then $\sigma = \tau$.
	\end{enumerate}
	This completes the proof. \hfill\proofSymbol

\section{Proof of \Cref{thm:main_convex_result} \label{sec:convexity}}
\subsection{Tools from convex optimization}\label{tools-convex}
	Let $\cX$ be a  real Hilbert space with inner product $\braket{\cdot,\cdot}$ and  $f: \cX\to \R \cup \{+\infty\}$ be a  function. In our applications, we  take $\cX = \Bsa(\H)$ equipped with the Hilbert-Schmidt inner product $\braket{\cdot,\cdot}_\text{HS}$. Let $\dom f = \{ x\in \cX: f(x) < \infty\}$ and assume $\dom f \neq \emptyset$. Let $\cX^* = \cB(\cX, \R)$ the set of bounded linear maps from $\cX$ to $\R$, equipped with the dual norm $\|\ell\|_* = \sup_{\|x\|=1} |\ell(x)|$ for $\ell \in \cX^*$. Since $\cX$ is a Hilbert space, by the Riesz-Fr\'echet representation, for each $\ell \in \cX^*$ there exists a unique $u_\ell \in \cX$ such that $\ell(x) = \braket{u_\ell, x}$ for all $x\in \cX$. We call $u_\ell$ the dual vector for $\ell$ (in particular the Hilbert-Schmidt dual in the case of $\cX = (\Bsa(\H), \braket{\cdot,\cdot}_\text{HS})$). We say $f$ is \emph{lower semicontinuous} if $\liminf_{x\to x_0} f(x) \geq f(x_0)$ for each $x_0\in \cX$. The \emph{directional derivative} of $f$ at $x\in \dom f$ in the direction $h\in \cX$ is given by
	\begin{equation} \label{eq:def_dir_deriv}
	f'(x;h) = \lim_{t \downarrow 0} \frac{1}{t}[f(x+th) - f(x)].
	\end{equation}
	If $f$ is convex, this limit exists in $\R\cup\{\pm \infty\}$. If the map $\phi_x(h):= \cX\ni h \mapsto f'(x;h)$ is linear and continuous, then $f$ is called \emph{G\^ateaux-differentiable} at $x\in \cX$. Moreover, if $f$ is G\^ateaux differentiable at every $x\in A\subset \cX$, then $f$ is said to be G\^ateaux-differentiable on $A$. If $f$ is convex and continuous at $x$, it can be shown that the map $\phi_x$ is finite and continuous. However, a continuous and convex function may not be G\^ateaux-differentiable. For example, $f:\R\to \R$, $f(x) = |x|$ has $\phi_0(h) = \lim_{t\downarrow 0} \frac{1}{t}[ f(x+th) - f(x)] = \lim_{t\downarrow 0} \frac{1}{t}|th| = |h|$ which is nonlinear. If $f$ is G\^ateaux-differentiable at $x$, we call the dual to $\phi_x$ as the \emph{G\^ateaux gradient} of $f$ at $x$, written $\nabla f (x)\in \cX$. That is, $\braket{\nabla f(x), h} = \phi_x(h)$ for each $h\in \cX$.

	The function $f$ is Fr\'echet-differentiable at $x$ if there is $y\in \cX$ such that
	\begin{equation}
	\lim_{\|r\|\to 0} \frac{1}{\|r\|} | f(x+r) - f(x) - \braket{y,r}| = 0 \label{eq:def_Frechet_deriv}
	\end{equation}
	and in this case, one writes $D f(x)\in \cX^*$ for the map $\cX\ni z\mapsto D f(x) z = \braket{y,z}$.
	 If $f$ is Fr\'echet-differentiable at $x$, then by taking $r = t h$ in \eqref{eq:def_Frechet_deriv} we find $\braket{y,h} = f'(x,h)$ and therefore $f$ is G\^ateaux-differentiable at $x$ with $\braket{\nabla f(x),h}  = D f(x)h$.

	By regarding differentiability as the existence of a linear approximation at a point, we can generalize it by defining a notion of a linear subapproximation at a point.
	\begin{definition}
	The \emph{subgradient} of a function $f:\cX\to \R$ at $x$ is the set
	\begin{equation}\label{def:subgradient}
	\partial f (x) = \{ u \in \cX: f(y) - f(x) \geq \braket{u,y-x} \,\,\forall y\in \cX \} \subset \cX.
	\end{equation}
	\end{definition}
	For a convex function $f$, the following properties hold:
	\begin{itemize}
		\item if $f$ is continuous at $x$, then $\partial f(x)$ is bounded and nonempty.
		\item if $f$ is G\^ateaux-differentiable at $x$, then $\partial f(x) = \{ \nabla f(x)\}$.
	\end{itemize}
	We briefly prove the second point here. Assume $f$ is G\^ateax-differentiable at $x$. Then
	\[
	\braket{\nabla f (x), y-x} = f'(x,y-x) = \lim_{t \downarrow 0} \frac{1}{t}[f((1-t)x+ty) - f(x)].
	\]
	By convexity, $f((1-t)x + ty) \leq (1-t) f(x) + t f(y)$. Therefore, $\braket{\nabla f (x), y-x} \leq f(y) - f(x)$; hence, $\nabla f(x)\in \partial f(x)$.

	On the other hand, given $u\in \partial f$, $h\in \cX$, and $t>0$, we can set $ y = th+x$. Then
	\[
	 \braket{u,th} = \braket{u,y-x} \leq f(y) -f(x)  = f(x+th) - f(x).
	\]
	Dividing by $t$ and taking the limit $t\downarrow 0$ yields $u = \nabla f(x)$.

	Fermat's Rule of convex optimization theory, stated below, provides a simple characterization of the minimum of a function in terms of the zeroes of its subgradient. Moreover, since
the subgradient of a G\^ateaux differentiable function consists of a single element, namely, its G\^ateaux gradient, finding its minimizer amounts to showing that its G\^ateaux gradient
is equal to zero.
	\begin{theorem}[Fermat's Rule] \label{thm:Global_minimizer_convex_fn}
	Consider a function $f: \cX\to \R\cup\{+\infty\}$ with $\dom f \neq \emptyset$. Then $\hat x\in \cX$ is a global minimizer of $f$ if and only if $0_{\cX} \in \partial f (\hat x)$, where $0_{\cX}$ is the zero vector of $\cX$.
	\end{theorem}
	\begin{proof}	
	$0_{\cX}\in \partial f(\hat x)$ if and only if 
	\[
	f(y) - f(\hat x) \geq \braket{0,y-\hat x} = 0
	\]
	for every $y \in \cX$,
	i.e. if and only if $\hat x$ is a global minimizer of $f$.
	\end{proof}

The following result proves useful in computing the subgradient of a sum of convex functions (see e.g. \cite[Theorem 3.30]{Peypouquet2015} for the proof).
\begin{theorem}[Moreau-Rockafellar] \label{thm:Moreau-Rockafellar}
Let $f,g: \cB(\H)\to \R\cup \{+\infty\}$ be convex and lower semicontinuous, with non-empty domains. For each $A\in \cB(\H)$, we have
\begin{equation}
\partial f(A) + \partial g(A) := \{a+b: a\in \partial f(A), b\in \partial g(A) \} \subset \partial(f+g)(A).
\end{equation}
Equality holds for every $A\in \cB(\H)$ if $f$ is continuous at some $A_0\in \dom(g)$.
\end{theorem}

\subsection{G\^ateaux-differentiable functions \label{sec:classF}}
	Let $\cF$ denote the class of functions $\varphi:\cD(\cH) \to \RR$ which are concave and continuous, and G\^ateaux-differentiable on $\cD_+(\cH)$. These include the following:
	\begin{itemize}
		\item The von Neumann entropy $\rho \mapsto S(\rho):= -\tr(\rho \log \rho)$. In \Cref{lem:Lrho_VN}, we show for $\rho>0$,
		\[
		\nabla S (\rho) = - \log \rho - \tfrac{1}{\log_{\mathrm{e}}(2)}\one.
		\]
		\item The conditional entropy $\rho_{AB} \mapsto S(A|B)_\rho := S(\rho_{AB}) - S(\rho_B)$, for which
		\[
		\nabla S(A|B)_{\cdot} (\rho_{AB})=  - (\log \rho_{AB} - \one_A\otimes \log \rho_B)
		\]
		as shown in \Cref{cor:Lrho_CE}. Note that for $\varphi(\rho) := S(A|B)_\rho$, the conditional entropy satisfies the interesting property that
		\begin{equation}
		\varphi(\rho) = \braket{\nabla \varphi (\rho),\rho}_\text{HS}.
		\end{equation}
		\item The $\alpha$-R\'enyi entropy for $\alpha \in (0,1)$. The map $\rho\mapsto S_\alpha(\rho)$ is concave for $\alpha\in (0,1)$, continuous, and has G\^ateaux gradient
		\[
		\nabla S_\alpha (\rho) = \frac{\alpha}{1-\alpha} \frac{1}{\tr [\rho^\alpha]} \rho^{\alpha-1}
		\]
		by \Cref{lem:grad_Renyi_ent}. 
		Note that the $\alpha$-R\'enyi entropy for $\alpha>1$ is not concave.
		\item The function $-T_\alpha$, for $\rho \mapsto T_\alpha(\rho):= \tr [\rho^\alpha]$ and $\alpha > 1$. This map is concave for $\alpha>1$, continuous, and by \Cref{lem:grad_Renyi_ent},  has G\^ateaux gradient
		 \[
		\nabla ({-T_\alpha})(\rho) = -\alpha \rho^{\alpha-1}.
		\]
		A state $\rho^* \in \Be(\sigma)$ maximizes $-T_\alpha$ over $\Be(\sigma)$ if and only if $\rho^*$ minimizes $T_\alpha$ over $\Be(\sigma)$. For $\alpha>1$, minimizing $T_\alpha$ is equivalent to maximizing $S_\alpha(\rho)$, using that $x\mapsto \frac{1}{1-\alpha}\log x $ is decreasing. Therefore, $\rho^*$ maximizes $-T_\alpha$ over $\Be(\sigma)$ if and only if it maximizes $S_\alpha$ over $\Be(\sigma)$. Thus, by considering the function $f(\rho)  =-T_\alpha(\rho)$ for $\alpha>1$ in \Cref{thm:main_convex_result}, one can establish conditions for $\rho^*$ to maximize the $\alpha$-R\'enyi entropy for $\alpha>1$.

	\end{itemize}

\subsection{A convex optimization result and its proof \label{sec:prove_convex_result}}
	\Cref{thm:main_convex_result} is a condensed form of the following theorem, whose proof we include below.
\begin{theorem} \label{thm:main_result}
	Let $f:\cD(\cH) \to \RR$ be a concave, continuous function which is G\^ateaux-differentiable on $\cD_+(\cH)$.  Moreover, given a state $\sigma\in \cD(\cH)$ and $\eps\in (0,1]$, for any $\rho\in \Bep(\sigma)$, 
	\begin{equation} \label{eq:main_result}
	\tr( \nabla f(\rho) (\rho-\sigma)) \leq \eps (\lambda_+(\nabla f(\rho)) - \lambda_-(\nabla f(\rho))). 
	\end{equation}
	Furthermore, the following are equivalent:

	\begin{enumerate}
		\item Equality in \eqref{eq:main_result},
		\item $\rho \in \argmax_{\Be(\sigma)} f$,
		\item 
		\begin{enumerate}
		\item Either $\frac{1}{2}\|\rho-\sigma\|_1= \eps$ or $\nabla f(\rho) = \lambda \one$ for some $\lambda \in \R$, and
		\item  we have 
		\[
		\pi_\pm \nabla f(\rho) \pi_\pm = \lambda_\pm(\nabla f(\rho)) \pi_\pm
		\] where $\pi_\pm$ is the projection onto the support of $\Delta_\pm$, and where $\Delta= \Delta_+ - \Delta_-$ is the Jordan decomposition of $\Delta:= \rho-\sigma$.
	\end{enumerate}
	\end{enumerate}

	\end{theorem}
We first prove the inequality \eqref{eq:main_result} by considering the Jordan decomposition of $\Delta:=\rho-\sigma$. Next, we convert the constrained optimization problem $\max_{\rho\in\Be(\sigma)} f(\rho)$ to an unconstrained optimization problem $\min h$ for a non-G\^ateaux differentiable function $h$ defined on all of $\Bsa(\H)$ by adding appropriate indicator functions for the sets $\cD(\cH)$ and $\{A\in \Bsa(\H): \frac{1}{2}\|A-\sigma\|_1 \leq \epsilon \}$.

 The Moreau-Rockafellar Theorem (\Cref{thm:Moreau-Rockafellar}) allows us to compute $\partial h(\rho)$ in terms of $\partial f(\rho) = \{\nabla f(\rho)\}$ and the subgradients of the indicator functions. We then show $0\in \partial h(\rho)$ if and only if equality is achieved in \eqref{eq:main_result}.

\paragraph{Proof of \Cref{thm:main_result}}
	We start with the following general result, which does not use convex optimization.

	\begin{lemma}
	Let $\eps>0$ and $\Delta\in\Bsa$ with $\tr (\Delta) = 0$, $\frac{1}{2}\|\Delta\|_1 \leq \eps$. Let $L\in \Bsa$. Then
	\begin{equation} \label{eq:UB_LDelta}
	 \tr( L \Delta ) \leq \eps (  \lambda_+(L)  -  \lambda_-(L) )
	\end{equation}
	with equality if and only if
		\begin{enumerate}
		\item Either $\frac{1}{2}\|\Delta\|_1= \eps$ or else $L = \lambda \one$ for $\lambda:= \lambda_+(L) = \lambda_-(L)$, and
		\item  we have 
		\[
		\pi_\pm L \pi_\pm = \lambda_\pm(L) \pi_\pm
		\] where $\pi_\pm$ is the projection onto the support of $\Delta_\pm$, and where $\Delta= \Delta_+ - \Delta_-$ is the Jordan decomposition of $\Delta$.
	\end{enumerate}
	\end{lemma}
	\begin{proof}
	Note $|\Delta| = \Delta_+ + \Delta_-$, and since $\Delta$ has trace zero, $\tr(\Delta_+) = \tr(\Delta_-)$. 
	We expand
	\[
	 \tr(L \Delta) = \tr(L \Delta_+)- \tr(L \Delta_-).
	\]
	Now, we use that $L$ is self-adjoint so that
	\[
	\lambda_-(L)\one \leq L \leq  \lambda_+(L)\one.
	\]
	Since $\Delta_\pm\geq 0$, we have
	\begin{align}	
	\tr(L \Delta_+) - \tr(L \Delta_-) 
	&\leq \tr(L \Delta_+) - \tr( \lambda_-(L)\Delta_-) \label{eq:lambda_min_ieq}\\
	&\leq \tr(\lambda_+(L) \Delta_+) - \tr( \lambda_-(L)\Delta_-) \label{eq:lambda_max_ieq}\\
	&= (\lambda_+(L) - \lambda_-(L) )\tr(\Delta_+), \nonumber
	\end{align}
	where the last line follows from the fact that $\tr(\Delta_+) =\tr(\Delta_-)$. However, $\tr(\Delta_+)\leq \eps$, so we have
	\begin{equation} \label{eq:epsilon_ieq}
	(\lambda_+(L) - \lambda_-(L) )\tr(\Delta_+)  \leq \eps (  \lambda_+(L)  -  \lambda_-(L) ).
	\end{equation}
	Thus, \eqref{eq:UB_LDelta} follows.
	To obtain equality in (\ref{eq:UB_LDelta}), we require equality in \eqref{eq:lambda_min_ieq}, \eqref{eq:lambda_max_ieq} and \eqref{eq:epsilon_ieq}. Equality in \eqref{eq:epsilon_ieq} is equivalent to condition $1$ in the statement of the lemma. We now show that equality in \eqref{eq:lambda_min_ieq} and \eqref{eq:lambda_max_ieq} is equivalent to condition $2$.

	Set $\lambda_+ = \lambda_+(L)$. Then since $\pi_+$ is the projection onto the support of $\Delta_+$, we have $\Delta_+ = \pi_+\Delta_+\pi_+$. Using cyclicity of the trace, we obtain
	\[
	\tr(L \Delta_+) = \tr(L \pi_+ \Delta_+\pi_+ ) = \tr(\pi_+ L \pi_+ \Delta_+).
	\]
	Equality in \eqref{eq:lambda_min_ieq} implies $\tr(L \Delta_+) = \lambda_+ \tr(\Delta_+)$, so
	\[
	\tr((\pi_+ L \pi_+)\Delta_+) = \tr(\lambda_+ \Delta_+)
	\]
	and thus
	\begin{equation}
	\tr(\Delta_+ (\lambda_+ \pi_+ - \pi_+ L \pi_+ )) = 0. \label{eq:tr_product_zero_proof}
	\end{equation}
	Since $\lambda_+$ is the largest eigenvalue of $L$ which is self-adjoint, we have $L \leq \lambda_+ \one$. Since conjugating by any operator preserves the semi-definite order, we have $\lambda_+ \pi_+ - \pi_+ L \pi_+ \geq 0$. Then since $\Delta_+$ restricted to $\pi_+$ is positive definite, \eqref{eq:tr_product_zero_proof} implies $\pi_+ L \pi_+ = \lambda_+\pi_+$.  Similarly, equality in \eqref{eq:lambda_max_ieq} implies $\pi_- L \pi_- = \lambda_-(L) \pi_-$.

	 Conversely, if $\pi_\pm L\pi_\pm = \lambda_\pm \pi_\pm$ for $\pm \in \{+,-\}$,  we immediately obtain equality in \eqref{eq:lambda_min_ieq} and \eqref{eq:lambda_max_ieq}.
	\end{proof}

	This lemma with the choices $\Delta = \rho-\sigma$ and $L \equiv L_\rho := \nabla f(\rho)$ gives the inequality \eqref{eq:main_result} of \Cref{thm:main_result}. It also gives the equivalence between the conditions $1$ and $3$ of the theorem.

	We now turn to the theory of convex optimization to establish the equivalence between conditions $1$ and $2$.
	Recall $f: \cD \to \R$ is a continuous, concave function which is G\^ateaux-differentiable on $\cD_+$. Let us write $\tilde{f} = -f$ which is convex, and note $L_\rho =- \nabla \tilde{f} (\rho)$. 
	With this choice, it remains to be shown that  $\rho^* \in \argmin_{\rho\in \Be(\sigma)} \tilde{f}(\rho)$ if and only if
	\begin{equation}
	\tr (L_{\rho^*} (\rho^* -\sigma)) \geq  \eps (\lambda_+(L_{\rho^*}) - \lambda_-(L_{\rho^*})). \label{eq:remains_to_show}
	\end{equation}

	 The Tietze extension theorem (e.g.~\cite[Theroem 2.2.5]{simon2015a}) allows one to extend a bounded continuous function defined on a closed set of a normal topological space (such as a normed vector space) to the entire space, while preserving continuity and boundedness.  We use this to extend $\tilde{f}$ (which is bounded as it is a continuous function on the compact set $\cD$) to the whole space $\Bsa$, using that $\cD$ is a closed set in $\Bsa$. We consider the closed and convex sets $\cD$ and
	\begin{equation}
	T := \left\{ A \in \Bsa: \|A - \sigma_{AB}\|_1 \leq 2 \eps \right\},\\
	\end{equation}
	and note $\Be = \cD\cap T$.
	We define $h: \Bsa\to \R\cup \{\infty\}$ as
	\[
	h =  	\tilde{f} + \delta_{\cD} + \delta_T
	\]
	where for $S\subset \Bsa$, the indicator function $\delta_S(A) = \begin{cases}
	0 & A\in S\\
	+\infty & \text{otherwise}.
	\end{cases}$

	We have the important fact that, by construction,
	\begin{equation}
	\argmin_{\omega\in\Be(\sigma)} \tilde{f}(\omega) = \argmin_{A\in \Bsa} h(A).
	\end{equation}

	By Theorem \ref{thm:Global_minimizer_convex_fn} $\hat A\in \Bsa$ is a global minimizer of $h$ if and only if $0\in \partial h(\hat A)$. Note each of $\tilde{f}, \delta_{\cD}, \delta_T$ is lower semicontinuous, convex, and has non-empty domain. Moreover, both $\tilde{f}$ and $\delta_T$ are continuous at any faithful state  $\omega \in \Bep(\sigma) \subset \dom \delta_{\cD}$. Hence, by \Cref{thm:Moreau-Rockafellar},
	\[
	\partial h = \partial \tilde{f} + \partial \delta_{\cD} + \partial \delta_T := \{\ell_f + \ell_{\cD} + \ell_T: \ell_f \in \partial \tilde{f},\, \ell_{\cD} \in \partial \delta_{\cD}, \,\ell_T \in \partial\delta_T\}.
	\]
	Hence, to obtain a complete characterization of
	\[
	\argmin h = \{A\in \Bsa: 0\in \partial h(A)\}
	\]
	we need to evaluate the subgradients of $\tilde{f}$, $\delta_T$, and $\delta_{\cD}$.

	Since $\tilde{f}$ is G\^ateaux-differentiable on $\cD_+$, for any $\omega\in \cD_+$, we have $\partial \tilde{f}(\omega) = \{-L_\omega\}$ for $L_\omega := -\nabla \tilde{f}(\omega)$.  The following two results (proven in \Cref{sec:proof_of_lemmas_convex}) describe the other two relevant subgradients.

	\begin{lemma} \label{lem:subgradT}
	For $A\in \Bsa$ with $0 < \frac{1}{2}\|A - \sigma_{AB}\|_1 \leq \eps$,
	\[
	\partial \delta_{T}(A)  =  \left\{ u \in \Bsa:  2 \eps \|u\|_\infty = \braket{u,A - \sigma_{AB}}  \right\}
	\]
	where
	\[
	\|u\|_\infty :=\lambda_+(|u|)= \sup_{A\in \Bsa,\, \|A\|_1=1} |\!\braket{u,A}\!|.
	\]
	\end{lemma}

	\begin{lemma} \label{lem:subgradD}
	Let $\omega \in \cD_+$. Then $\partial \delta_{\cD}(\omega) = \left\{ x \one: x\in \R \right\}$.
	\end{lemma}

	Putting together these results, we have, for $\rho^* \in \Bep(\sigma)$,
	\[
	0 \in \partial h(\rho^*) \iff 0 = -L_{\rho^*} + G + x \one
	\]
	for some $G\in \partial \delta_T({\rho^*})$ and $x\in\R$, where $G$ satisfies
	\[
	2 \eps \|G\|_\infty=  \tr(G (\rho^* - \sigma)).
	\]
	Then $G =  L_{\rho^*} - x \one$, which implies that
	\[
	\tr( L_{\rho^*} (\rho^* - \sigma)) = 2 \eps \|L_{\rho^*} - x \one\|_\infty .
	\]
	Set $\alpha = \frac{\tr( L_{\rho^*} (\rho^* - \sigma))}{2 \eps}.$
	Note $\alpha$ and $L_{\rho^*}$ depend on $\rho^*$. Then we have
	\[
	0\in \partial h(\rho^*) \iff \exists\,  x\in \R: \alpha = \|L_{\rho^*}-x \one\|_\infty.
	\]
	Since $\|L_{\rho^*}-x \one\|_\infty = \max_{\lambda \in \spec L_{\rho^*}} |\lambda- x|$, we have that $0\in \partial h({\rho^*})$ if and only if
	\begin{equation}
	\exists\, x\in \R: \alpha =  \max_{\lambda \in \spec L_{\rho^*}} |\lambda- x|. \label{eq:cond_x}
	\end{equation}
	Let 
	\begin{equation}
	 q(x) = \max_{\lambda \in \spec L_{\rho^*}} |\lambda - x| \label{eq:def_qx}
	 \end{equation} 
	 We immediately see that $q$ is continuous, $q(0) = \lambda_+(|L_{\rho^*}|)$, and $\lim_{z\to\pm\infty} q(z) = +\infty$. In fact, we can write a simple form for $q(x)$ as the following lemma, which is proven in \Cref{sec:proof_of_lemmas_convex}, shows. Set $\lambda_+ \equiv \lambda_+(L_{\rho^*})$ and $\lambda_- \equiv \lambda_-(L_{\rho^*})$ in the following.
	\begin{lemma} \label{lem:q(x)}
	 The quantity $q(x)$ defined by \eqref{eq:def_qx} can be expressed as follows:
	\[
	q(x) = \frac{\lambda_+ - \lambda_-}{2} + \left| \frac{\lambda_+ + \lambda_-}{2} - x \right|.
	\]
	\end{lemma}

	Thus, Lemma \ref{lem:q(x)} implies that the range of the function $q$ is $[\frac{\lambda_+ - \lambda_-}{2}, \infty)$. Hence, \eqref{eq:cond_x} holds if and only if
	\[
	 \frac{\lambda_+ - \lambda_-	}{2} \leq \alpha.
	\]
	Substituting  $\alpha = \frac{\tr( L_{\rho^*} (\rho^* - \sigma))}{2 \eps}$ in the above expression yields \eqref{eq:remains_to_show} and therefore concludes the proof of Theorem \ref{thm:main_result}.\hfill\proofSymbol

	\bigskip

	Below, we collect some results (proven in \Cref{sec:proof_of_lemmas_convex}) relating to the G\^ateaux gradients of relevant entropic functions which are candidates for $f$ in \Cref{thm:main_result}.
	\begin{lemma}\label{lem:Lrho_VN}
	The von Neumann entropy $\rho\mapsto S(\rho):=-\tr (\rho \log \rho)$ is G\^ateaux-differentiable at each $\rho >0$ and $\nabla S(\rho) = - \log \rho - \tfrac{1}{\log_{\mathrm{e}}(2)}\one$.
	\end{lemma}

	\begin{corollary}\label{cor:Lrho_CE}
	The conditional entropy $\rho_{AB}\mapsto S(A|B)_\rho:=S(\rho_{AB}) - S(\rho_B)$ is G\^ateaux-differentiable at each $\rho_{AB} >0$ and $\nabla S(A|B)_{\cdot}(\rho_{AB}) =-(\log \rho_{AB} - \one_A\otimes \log\rho_B)$.
	\end{corollary}

	\begin{lemma}\label{lem:grad_Renyi_ent}

	For $\alpha\in (0,1)\cup (1,\infty)$, the map $T_\alpha$ and the R\'enyi entropy $S_\alpha$ are G\^ateaux-differentiable at each $\rho \in \cD_+$, with G\^ateaux gradients
	\[
	\nabla T_\alpha(\rho) = \alpha \rho^{\alpha-1}, \qquad
	 \nabla S_\alpha (\rho) = \frac{\alpha}{1-\alpha}\frac{1}{\tr [\rho^\alpha]} \rho^{\alpha-1}.
	\]
	\end{lemma}
	\subsection{Proofs of the remaining lemmas and corollaries \label{sec:proof_of_lemmas_convex}}

	\paragraph{Proof of \Cref{cor:main_thm_for_CE}}
	We write the conditional entropy (defined in \cref{cond-vN}) as the map 
	\begin{equation}
		 \begin{aligned}
	 S(A|B)_\cdot\quad: \quad &\cD(\H_{AB}) \to \R\\
	 & \rho \mapsto S(A|B)_\rho.
	 \end{aligned} \label{eq:CE_as_map}
	\end{equation}
	 By \Cref{cor:Lrho_CE}, for any $\rho_{AB}>0$ we have 
	$$\nabla S(A|B)_{\cdot} (\rho) =  - (\log \rho_{AB} - \one_A\otimes \log \rho_B) = - L_\rho.$$ 
	 Substituting this in the left-hand side of \eqref{eq:main_result}, we obtain
	\begin{align*}	
	\tr (L_\rho (\sigma_{AB} - \rho_{AB})) &=\tr[\sigma_{AB} \log \rho_{AB}] -\tr[\rho_{AB}\log \rho_{AB}] - \tr[\sigma_B\log \rho_B] + \tr[\rho_B\log \rho_B]\\
	&= S(A|B)_\rho + \tr[\sigma_{AB} \log \sigma_{AB}] - \tr[\sigma_{AB} (\log \sigma_{AB} - \log \rho_{AB})] \\
	&\qquad\qquad\quad - \tr[\sigma_B\log\sigma_B] + \tr[\sigma_B(\log \sigma_B - \log \rho_B)]\\
	&= S(A|B)_\rho - S(A|B)_\sigma - D(\sigma_{AB}\|\rho_{AB}) + D(\sigma_B\|\rho_B).
	\end{align*}
	The statement of \Cref{cor:main_thm_for_CE} then follows from \Cref{thm:main_result}.\hfill\proofSymbol

	\paragraph{Proof of Lemma \ref{lem:subgradT}}

	Let $\cA := \left\{ u \in \Bsa:  2 \eps \|u\|_\infty = \braket{u, A - \sigma_{AB}})  \right\}$ and set 
	\[
	C :=\{ B\in \Bsa: \| B -\sigma_{AB} \|_1 \leq \eps\}\subset \Bsa.
	\]
	Let $u \in \cA$. Then for $y\in C$,
	\begin{align*}	
	\braket{u,y - A} &= \braket{u,y - \sigma_{AB}} + \braket{u,\sigma_{AB} - A} = \braket{u,y - \sigma_{AB}} - 2 \eps \|u\|_\infty \\
	&\leq \|u\|_\infty \|y-\sigma_{AB}\|_1 - 2 \eps \|u\|_\infty \\
	&= \|u\|_\infty (\|y-\sigma_{AB}\|_1 - 2 \eps) \leq 0
	\end{align*}
	and thus $u \in \partial \delta_T(A)$. 
	On the other hand, take $u \in \partial \delta_T(A)$. Then
	\[
	\|u\|_\infty = \sup_{\|x\|_1 = 1}|\!\braket{u,x}\!|
	\]
	is achieved at some $x^* \in \Bsa$ since the set $\{x\in \Bsa: \|x\|_1 = 1\}$ is compact, using that $\Bsa$ is finite-dimensional. Then $y = \sigma_{AB} + 2\eps x \in C$. Hence,
	\[
	0 \geq \braket{u,y-A} = \braket{u, \sigma_{AB} - A} + 2 \eps \braket{u,x} = \braket{u,\sigma_{AB} - A} + 2 \eps \|u\|_\infty
	\]
	and thus $\braket{u,A- \sigma_{AB}}  \geq   2 \eps \|u\|_\infty$. 
	Then by the bound
	\[
	\braket{u,A-\sigma_{AB}} \leq \|u\|_\infty \|A-\sigma_{AB}\|_1 \leq 2 \eps \|u\|_\infty
	\] we have $u \in \cA$.\hfill\proofSymbol

\paragraph{Proof of Lemma \ref{lem:subgradD}}

	By definition,
	\begin{align*}	
	\partial \delta_{\cD}(\omega) &= \left\{ u\in \Bsa: \braket{u, y-\omega} \leq 0 \, \text{ for all }y\in \cD  \right\}.
	\end{align*}
	Since $u$ is self-adjoint, we can write its eigen-decomposition as $u = \sum_{k=1}^d \alpha_k \ketbra{k}{k}$. If $\alpha_k = \alpha_j$ for each $j,k$, then $u \propto \one$. Conversely, $\alpha_k \one \in \partial \delta_{\cD}(\omega)$ since $\tr[\alpha_k \one (y-\omega)] = \alpha_k (\tr(y)  - \tr (\omega)) = 0$.

	Otherwise, assume $\alpha_k > \alpha_j$ for some $k,j\in \{1,\dotsc, d\}$.
	Let
	\[
	y := \sum_{i, \, i\neq k,\, i \neq j} \braket{i| \omega| i} \ketbra{i}{i} + ( \braket{k| \omega| k} + \braket{j| \omega| j} )\ketbra{k}{k} \in \Bsa.
	\]
	Then $\tr (y) = \sum_i \braket{i| \omega| i}  = \tr (\omega) = 1$, and $y\geq 0$ since all of its eigenvalues are non-negative. Note $\ket{j}$ is an eigenvector of $y$ with eigenvalue zero. Next,
	\begin{align*}	
	\tr(u ( y -\omega)) &= \sum_i \braket{i| (y-\omega) |i} \\
	&= \sum_{i, \, i\neq k,\, i \neq j} (\braket{i| \omega| i} - \braket{i|\omega| i}) \\
	&\qquad +\alpha_k ( \braket{k| \omega| k} + \braket{j| \omega| j}  - \braket{k|\omega| k}) + \alpha_j ( 0 -\braket{j|\omega| j})\\
	&=(\alpha_k - \alpha_j)\braket{j|\omega| j} > 0.
	\end{align*}
	Thus, $u \not \in \partial \delta_{\cD} (\omega)$.\hfill\proofSymbol

	\paragraph{Proof of Lemma \ref{lem:q(x)}}
	First, we establish
	\begin{equation}
	q(x) = \max(|\lambda_+- x|,|\lambda_--x|). \label{eq:formq1}
	\end{equation}
	Given $\lambda\in \spec L$ and $x\in \R$, if $x < \lambda$, then 
	\[
	|x- \lambda| = \lambda - x \leq \lambda_+ - x = |\lambda_+-x|
	\]
	and otherwise
	\[
	|x- \lambda| = x-\lambda  \leq x -\lambda_- = |\lambda_--x|,
	\]
	yielding \eqref{eq:formq1}.
	Next, set $r := \frac{\lambda_+ - \lambda_-}{2}$ and $m := \frac{\lambda_+ + \lambda_-}{2}$. Then $\lambda_\pm =  m \pm r$. Therefore, for $x\in \R$,
	\[
	|\lambda_\pm - x| = |m \pm r -x| \leq r + |m - x|,
	\]
	yielding $q(x) = r + |m-x|$ as claimed. \hfill\proofSymbol

	\paragraph{Proof of Lemma \ref{lem:Lrho_VN}}

	Let us introduce the Cauchy integral representation of an analytic function. 
	If $g$ is analytic on a domain $G\subset \C$ and $A$ is a matrix with $\spec A \subset G$, then we can write
	\[
	g(A) = \frac{1}{2\i \pi} \int_\Gamma g(z) (z \one - A)\inv \d z,
	\]
	and
	\begin{equation}
	g'(A)= \frac{1}{2\i \pi} \int_\Gamma g(z) (z \one - A)^{-2} \d z, \label{eq:gprimeA}
	\end{equation}
	where $\Gamma\subset G$ is a simple closed curve with $\spec A\subset \Gamma$, and $g': G\to \C$ is the derivative of $g$. \cite{STICKEL1987} shows that with these definitions, the Fr\'echet derivative of $g$ at $A$ exists, and when applied to a matrix $X$ yields
	\begin{equation*}
	D(g)(A)X = \frac{1}{2\i \pi}\int_\Gamma g(z) (z \one - A)\inv X (z\one - A)\inv \d z. \label{eq:DgA_integral_rep}
	\end{equation*}
	Therefore, by cyclicity of the trace,
	\begin{equation} \label{eq:trace-Frechet-deriv-formula}
	\tr(D(g)(A)X) = \tr \left[\frac{1}{2\i \pi}\int_\Gamma g(z) (z \one - A)^{-2} \d z\, X\right] = \tr(g'(A) X)
	\end{equation}
	using \eqref{eq:gprimeA} for the second equality.

	We may write consider von Neumann entropy as the map
	\[
	S: \cD \to \R, \qquad \rho \mapsto S(\rho) = - \tr[ \rho \log \rho].
	\]
	Then we may write $S = \tr \circ \eta$ for $\eta(x) = - x \log x$ which is analytic on $\{\zeta\in \C: \re \zeta > 0 \}$, with derivative $\eta'(x) = - \log x - \frac{1}{\log_{\mathrm{e}}(2)}$. Then for $\rho > 0$,
	\[
	D(S)(\rho) = D(\tr \circ \eta)(\rho) = D(\tr)(\eta(\rho)) \circ D(\eta)(\rho).
	\]
	by the chain rule for Fr\'echet derivatives. By the linearity of the trace, $D(\tr)B(X) = \tr(X)$ for any $X,B\in \Bsa$. So for $X\in \Bsa$,
	\[
	D(S)(\rho) X  = \tr[D(\eta)(\rho)X].
	\]
	By \eqref{eq:trace-Frechet-deriv-formula} for $g = \eta$, we have
	\begin{equation*}
	D(S)(\rho)X = \tr[\eta'(\rho) X] = \tr [ (- \log \rho - \tfrac{1}{\log_{\mathrm{e}}(2)}\one)X].\tag*{\proofSymbol}
	\end{equation*}

	\paragraph{Proof of Corollary \ref{cor:Lrho_CE}}

	We decompose the conditional entropy map (\Cref{eq:CE_as_map}) by writing
	\[
	S(A|B)_\cdot = S_{AB}(\cdot) - S_B \circ \tr_A (\cdot): \D(H_{AB})\to \R
	\]
	where we have indicated explicitly the domain of each function in the subscript. That is, $S_{AB}: \D(\H_{AB})\to \R$ is the von Neumann entropy on system $AB$, and $S_B:\D(\H_B)\to \R$ is the von Neumann entropy on $B$.

	 Since $\tr_A: \Bsa(\H_{AB})\to \Bsa(\H_B)$ is a bounded linear map and $S_B$ concave and continuous, the chain rule for the composition with a linear map (see Prop. Prop.~3.28 of \cite{Peypouquet2015}) gives 
	 \[
	  \nabla (S_B\circ \tr_A)(\omega_{AB}) = \tr_A^* \circ \nabla S_B  (\omega_{B})= \one_A \otimes \nabla S_B(\omega_B)
	  \]
	   for any $\omega_{AB}\in \D(\H_{AB})$, where $\tr_A^*$ is the dual to the map $\tr_A$.
	As $\nabla S_B(\omega_B) = - (\log \omega_B+\tfrac{1}{\log_{\mathrm{e}}(2)} \one_B)$, we have
	\begin{align*}	
	\nabla S(A|B)_\cdot (\omega_{AB}) &= \nabla S_{AB} (\omega_{AB}) - \nabla (S_{B}\circ \tr_A)(\omega_{AB})\\
	  &= - (\log \omega_{AB} + \tfrac{1}{\log_{\mathrm{e}}(2)}\one_{AB})+ \one_A\otimes (\log \omega_B + \tfrac{1}{\log_{\mathrm{e}}(2)}\one_B)\\
	  &= -(\log \omega_{AB} -\one_A\otimes \log \omega_B). \tag*{\proofSymbol}
	\end{align*}

	\paragraph{Proof of Lemma~\ref{lem:grad_Renyi_ent}}

	The $\alpha$-R\'enyi entropy $\alpha\in(0,1)\cup(1,\infty)$ can be described by the map
	\[
	S_\alpha: \cD \mapsto \R, \qquad \rho \mapsto S_\alpha(\rho) :=  \frac{1}{1-\alpha} \log \tr (\rho^\alpha).
	\]
	Let us use the notation $\pow_\alpha : \Bsa\to \Bsa$, $A\mapsto \pow_\alpha(A) := A^\alpha$. Then
	\[
	T_\alpha= \tr \circ \pow_\alpha
	\]
	and $(1-\alpha) S_\alpha = \log \circ T_\alpha$.
	For $\rho \in \cD_+$, the function $\pow_\alpha$ is analytic on $\{\zeta\in \C: \re \zeta > 0\}$. As in the proof of \Cref{lem:Lrho_VN} then, using the chain rule and linearity of the trace we find
	\begin{equation*}	
	D(T_\alpha)(\rho) = \tr [D(\pow_\alpha)(\rho) X].
	\end{equation*}
	Invoking \eqref{eq:trace-Frechet-deriv-formula} for $g(z) = z^\alpha$, with $\Gamma \subset \{\zeta\in \C: \re \zeta > 0\}$  a simple closed curve enclosing the spectrum of $\rho$, we have
	\begin{equation}
	D(T_\alpha)(\rho)X = \alpha \tr (\rho^{\alpha-1} X).
	\end{equation}
	Moreover,
	\begin{align*}	
	D(S_\alpha)(\rho)X &= \frac{1}{1-\alpha}D(\log \circ T_\alpha) (\rho)\\
	&= \frac{1}{1-\alpha} D(\log)(\tr (\rho^\alpha)) \circ D(T_\alpha)(\rho) \\
	&=  \frac{1}{1-\alpha}\frac{\alpha}{\tr [\rho^\alpha]} \tr(\rho^{\alpha-1 }X). \tag*{\proofSymbol}
	\end{align*}

\section{Conclusion and Open Questions}
Given an arbitrary quantum state $\sigma$, we obtain local continuity bounds for the von Neumann entropy 
and its $\alpha$-R\'enyi entropy (with $\alpha \in (0,1) \cup (0,\infty)$). Our bounds are local in the sense that 
they depend on the spectrum of $\sigma$. They are seen to be stronger than the respective global continuity bounds, 
namely the Audenaert-Fannes bound for the von Neumann entropy, and the known bounds on the $\alpha$-R\'enyi entropy \cite{Renyi-CMNF,Ras2011},  cf.~\Cref{fig:local-VN-vs-AF,fig:local-Renyi-vs-global}.

We obtain these bounds by fixing $\eps \in (0,1)$ and explicitly constructing states $\rho^*_\eps(\sigma)$ and 
$\rho_{*,\eps}(\sigma)$ which, respectively, have maximum and minimum entropies among all states which lie in an $\eps$-ball (in trace distance)
around $\sigma$. These states satisfy a majorization order, and the minimum (resp.~maximum) entropy state
is the maximum (resp.~minimum) state in this order, consistent with the Schur concavity
of the entropies considered. The state $\rho_{*,\eps}(\sigma)$ lies on the boundary of the $\eps$-ball, as does
the state $\rho^*_\eps(\sigma)$, unless $\rho^*_\eps(\sigma)$ is the completely mixed state. The state $\rho^*_\eps(\sigma)$ 
is also the unique maximum entropy state. Moreover,
it has certain interesting properties: it satisfies a semigroup property, and saturates the triangle inequality for the completely mixed
state and $\sigma$ (cf.~\Cref{thm:max-min-ball}). The explicit form of these states also allows us to obtain expressions
for smoothed min- and max- entropies, which are of relevance in one shot information theory. 

Our construction of the maximum entropy state is described and motivated by a more general mathematical
result, which employs tools from convex optimization theory (in particular Fermat's Rule), and might be of independent interest. 
To state the result, we introduce the notion of G\^ateaux differentiability, which can be viewed as an extension of
directional differentiability to arbitrary real vector spaces. Examples of G\^ateaux differentiable functions include
the von Neumann entropy, the conditional entropy and the $\alpha$-R\'enyi entropy. Our result (\Cref{thm:main_convex_result}) provides
a necessary and sufficient condition under which a state in the $\eps$-ball maximizes a concave 
G\^ateaux-differentiable function. 
Even though we consider optimization over the $\eps$-ball in trace distance in this paper, the techniques used to prove \Cref{thm:main_convex_result} may be able to be extended to other choices of the $\eps$-ball (e.g.~over the so-called purified distance).

The majorization order, $\rho^*_\eps(\sigma) \prec \rho \prec \rho_{*,\eps}(\sigma)$ $\, \forall \, \rho \in B_\eps(\sigma)$,
established in \Cref{thm:max-min-ball}  has interesting implications regarding the transformations of bipartite purification of these states under local operations and classical
communication (LOCC), via Nielsen's majorization criterion \cite{Nielsen-LOCC}, which we intend to explore in future work.

\section*{Acknowledgements}
We are grateful to Mark Wilde for suggesting an interesting problem which led us to the problems studied in this paper, and to Andreas Winter for his suggestion of looking at local continuity bounds. We would also
like to thank Koenraad Audenaert, Barbara Kraus, Sam Power, Cambyse Rouz\'e and Stephanie Wehner for helpful discussions. We would particularly like to thank Hamza Fawzi for insightful discussions about convex optimization.

\printbibliography
\end{document}